\documentclass[11pt]{article}
\pdfoutput=1
\usepackage{jheppub}
\usepackage{amsmath}
 \usepackage{bm}
\usepackage{setspace}
\usepackage{amsthm}
\newtheorem{theorem}{Theorem}
\newtheorem{lemma}{Lemma}

\usepackage{hyperref}
\usepackage{subcaption}

\newcommand{\ca}{\cos(\vec{\alpha}_a \cdot \vec{\sigma})}
\newcommand{\cN}{\cos(\vec{\alpha}_N \cdot \vec{\sigma})}
\newcommand{\sa}{\sin(\vec{\alpha}_a \cdot \vec{\sigma})}

\newcommand{\sN}{\sin(\vec{\alpha}_N \cdot \vec{\sigma})}
\newcommand{\ala}{\vec{\alpha}_a}
\newcommand{\alb}{\vec{\alpha}_b}
\newcommand{\alN}{\vec{\alpha}_N}
\newcommand{\vs}{\vec{\sigma}}
\newcommand{\vp}{\vec{\phi}}
\newcommand{\sumaN}{\sum_{a=1}^N}
\newcommand{\sumbN}{\sum_{b=1}^N}
\newcommand{\cjwj}{i2 \pi \sum_{j=1}^{N-1} c_j \vec{w}_j}
\newcommand{\xinf}{\vec{x}(\infty)}
\newcommand{\xninf}{\vec{x}(-\infty)}
\newcommand{\sumcj}{\sum_{j=1}^{N-1} c_j}

\newcommand{\xv}{\vec{x}}
\newcommand{\wv}{\vec{w}}
\newcommand{\phiv}{\vec{\phi}}
\newcommand{\sigmav}{\vec{\sigma}}
\newcommand{\ava}{\vec{\alpha}_a}
\newcommand{\alphav}{\vec{\alpha}}
\newcommand{\rhov}{\vec{\rho}}

\usepackage[font=scriptsize]{caption}
\def\R{{\mathbb R}}
\def\S{{\mathbb S}}
\def\Z{{\mathbb Z}}

\def\beq{\begin{equation}}
\def\eeq{\end{equation}}

\def\R{{\mathbb R}}
\def\S{{\mathbb S}}
\def\Z{{\mathbb Z}}

\def\beq{\begin{equation}}
\def\eeq{\end{equation}}

\title{Domain walls and deconfinement: a semiclassical picture of discrete anomaly inflow}
\author{Andrew A. Cox, Erich Poppitz, Samuel S.Y. Wong}
\affiliation{Department of Physics,   University of Toronto, 
Toronto, ON M5S 1A7, Canada}
\emailAdd{aacox@physics.utoronto.ca}\emailAdd{poppitz@physics.utoronto.ca}  \emailAdd{samuelsy.wong@mail.utoronto.ca}    
 \abstract{We study the  physics of quark deconfinement on domain walls in  four-dimensional supersymmetric $SU(N)$ Yang-Mills theory, compactified on a small circle with supersymmetric boundary conditions. 
We numerically examine the properties  of BPS domain walls connecting  vacua $k$ units apart. We also determine their  electric fluxes  and use the results to show that Wilson loops of any nonzero $N$-ality exhibit  perimeter law on all   $k$-walls. Our results confirm and extend, to  all $N$ and $k$, the validity of the semiclassical picture of deconfinement of Anber, Sulejmanpasic  and one of us (EP), \href{https://arxiv.org/abs/1501.06773}{arXiv:1501.06773}, providing a microscopic explanation of mixed $0$-form/$1$-form anomaly inflow.}

\begin{document}

\maketitle

\section{Introduction, Summary, and Outlook}

Supersymmetric Yang-Mills (SYM) theory has been the subject of many studies over the years,   due to its tractability, facilitated by supersymmetry, and to its similarity  to nonsupersymmetric pure YM theory. Of special interest to us is  the theory formulated on  $\R^{3} \times \S^1$, with supersymmetric boundary conditions on $\S^1$ \cite{Seiberg:1996nz,Aharony:1997bx}. The dynamics is drastically simplified in the small-circle limit $L N \Lambda \ll 1$, where $L$ is the $\S^1$ circumference, $N$ is the number of colors, and $\Lambda$ is its strong coupling scale \cite{Davies:1999uw,Davies:2000nw} (we consider  $SU(N)$ gauge groups only).
Remarkably, center stability, confinement, and discrete chiral symmetry breaking  become intertwined and are all due to the proliferation of various
topological molecules in the Yang-Mills vacuum: magnetic \cite{Unsal:2007jx,Unsal:2007vu} and neutral ``bions" \cite{Poppitz:2011wy,Poppitz:2012sw,Argyres:2012ka,Argyres:2012vv,Poppitz:2012nz}, composite objects made of various monopole-instantons \cite{Lee:1997vp,Kraan:1998pm}. The magnetic bions, in particular, provide a remarkable locally four-dimensional  generalization  \cite{Unsal:2007jx,Unsal:2007vu}  of the Polyakov mechanism \cite{Polyakov:1976fu} of confinement: despite looking three-dimensional at long distances, the theory remembers much about its four-dimensional origin. In many cases, it is known or believed that the small-$L$ theory is connected to the $\R^4$ theory ``adiabatically," i.e. without a phase transition. We also note that  lattice studies \cite{Bergner:2015cqa,Bergner:2018unx,Bergner:2019dim} have already provided evidence for the adiabaticity. 

The $\R^{3} \times \S^1$ setup described above provides a rare example of  analytically tractable nonperturbative phenomena in four dimensions and many of its aspects have been the subject of previous investigations, see \cite{Dunne:2016nmc} for a review.

\subsection{Motivation}

The purpose of this paper is to continue and complete the study of domain walls (DW) of Ref.~\cite{Anber:2015kea}. There,   using semiclassical physics, it was shown that heavy fundamental quarks, represented by unit $N$-ality Wilson loops,  are deconfined on the DWs of SYM and, hence, that confining strings can end on  DWs.\footnote{This phenomenon has previously been noted using a string (M)-theory embedding of SYM \cite{SJRey:1998,Witten:1997ep} and argued for using large-$N$ arguments,  e.g.~\cite{Armoni:2003ji}.}
The physical picture of
 \cite{Anber:2015kea} is that in the semiclassical limit on $\R^{3} \times \S^1$, static DWs  are lines (in $\R^2$) of finite energy per unit length. Along their length, the DWs carry  a  fraction of the chromoelectric flux of quarks. These fractional fluxes are precisely such  that a pair of DWs form a ``double string," shown on Fig.~\ref{fig:01},  stretched between a quark and antiquark,   leading to quark confinement with a linearly rising potential. The double-string   picture further implies that quarks become deconfined on DWs---the quark's chromoelectric flux is absorbed by the DWs, of equal tension,   to the left and the right. The equal tensions of different walls here is due to supersymmetry, as many of the walls are Bogomolnyi-Prasad-Sommerfield (BPS) protected objects, as well as to unbroken 0-form center symmetry. This physical picture of deconfinement on walls is illustrated on Fig.~\ref{fig:02} for an $SU(2)$ gauge group.

\begin{figure}[t]
\begin{subfigure}[t]{.47  \textwidth}
  \includegraphics[width= 1 \textwidth]{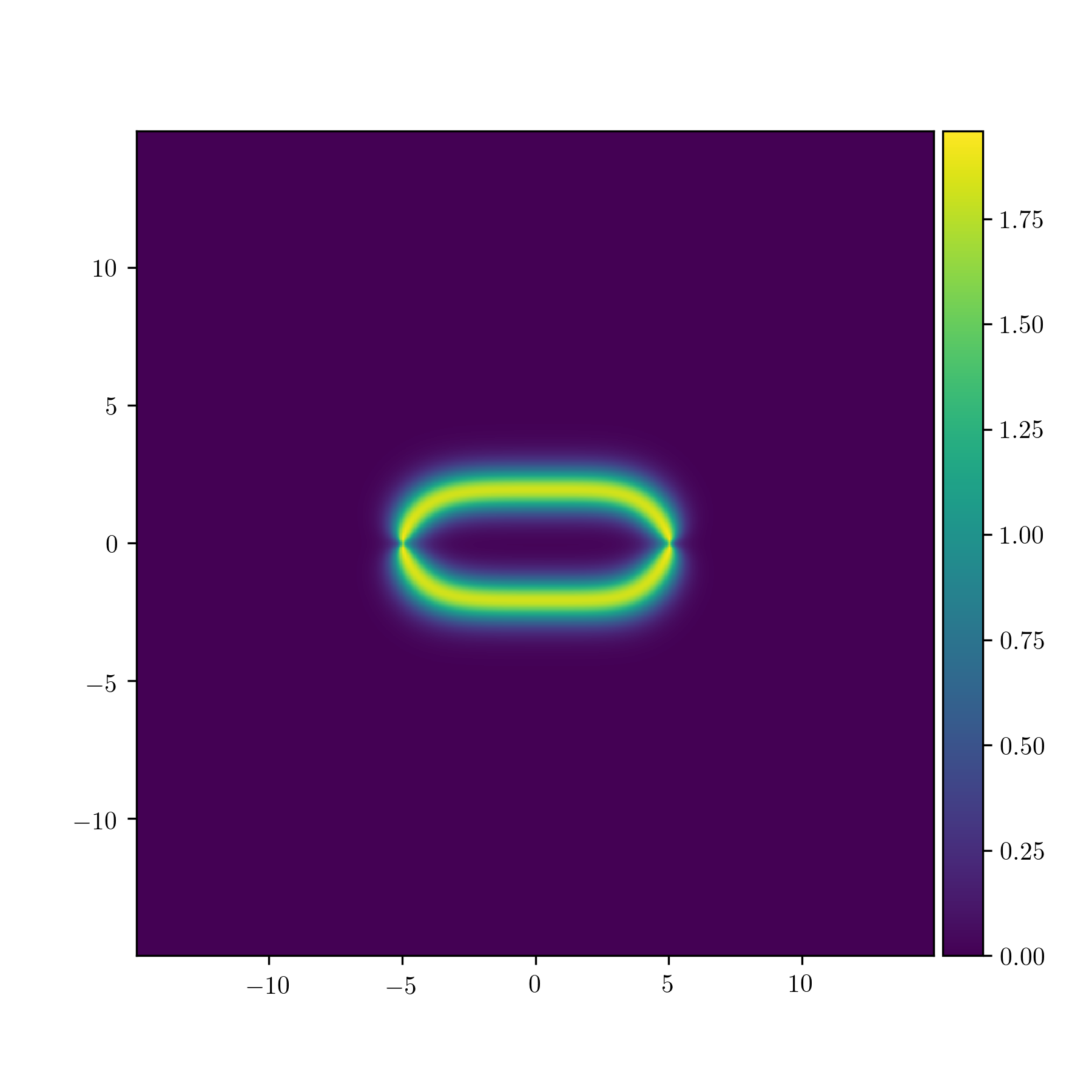}
  \caption{The energy density of a ``double-string" confining  configuration composed of two degenerate BPS DWs in $SU(3)$ SYM. The quark and antiquark have weights $\pm \pmb w_1$ of the fundamental representation. The string is embedded in vacuum 1, see (\ref{vacua}), while inside the double-string the fields have the values of vacuum  0. Distances are measured in units of the Compton wavelength of the heaviest dual photon. Similar double string configurations confine fundamental quarks for any number of colors.}
  \label{fig:01}
\end{subfigure} 
\qquad
\begin{subfigure}[t]{.47 \textwidth}
  \includegraphics[width= 1\textwidth]{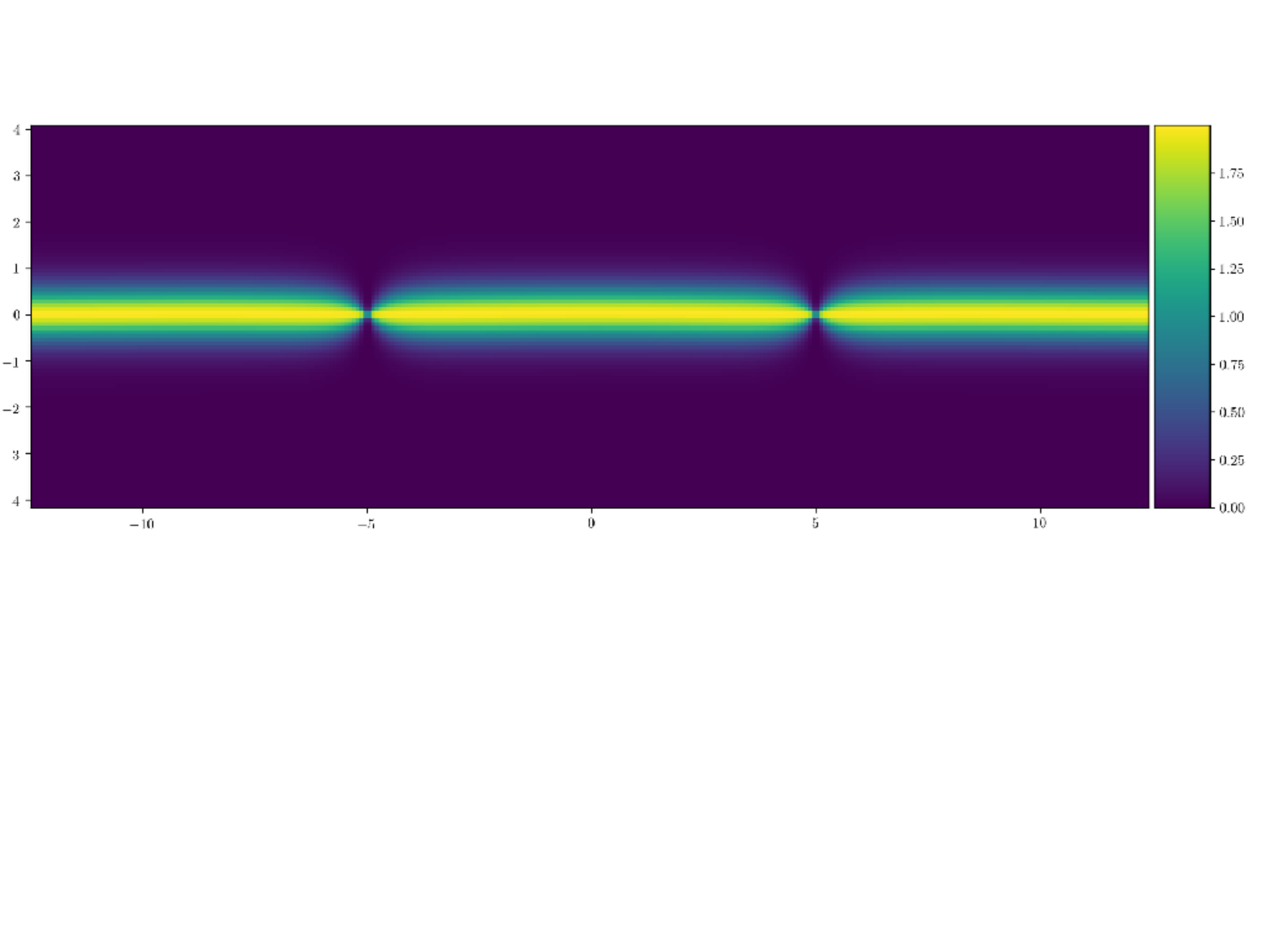}
  \caption{Deconfinement of a quark/antiquark pair on the DW, shown here for $SU(2)$ SYM. This configuration can be thought of as the double string configuration on the left ``opened up." Vacuum 0 is on the top and vacuum 1---on the bottom. As the tensions of the BPS DWs absorbing a quark's electric flux   are equal, there is no distance dependence of the quark-antiquark pair's energy, as shown on  Fig.~\ref{fig:deconfinement_energy}. As in the figure on the left, the plot here shows the energy density.}
  \label{fig:02}
\end{subfigure}
\caption{Confinement in the bulk and deconfinement on the wall. Section \ref{dwbaryons}  explains  how these pictures are obtained.}
\end{figure}

We note in passing that the above ``double-string" picture of confining strings also holds in nonsupersymmetric deformed Yang Mills theory \cite{Unsal:2008ch} at $\theta=\pi$, as well as in nonsupersymmetric adjoint QCD on $\R^3 \times \S^1$, as explained in \cite{Anber:2015kea}.

A recent parallel development is the realization that deconfinement of quarks on the DWs in SYM is a manifestation of    ``discrete anomaly inflow," due to the newly discovered mixed $0$-form/$1$-form symmetry discrete  't Hooft anomalies \cite{Gaiotto:2014kfa,Gaiotto:2017yup,Gaiotto:2017tne}. The mixed 0-form/1-form anomalies  imply that in theories with such anomalies, DWs occurring due to 0-form discrete symmetry breaking have a nontrivial structure on their worldvolume. Such DWs  have recently received some attention \cite{Sulejmanpasic:2016uwq,Komargodski:2017smk,Draper:2018mpj,Ritz:2018mce,Anber:2018xek,Anber:2018jdf,Nishimura:2019umw,Misumi:2019dwq,Cordova:2019jnf,Cordova:2019uob}. The nontrivial DW physics is usually described in terms of a topological quantum field theory (TQFT) living on the DW, e.g. \cite{Acharya:2001dz,Gaiotto:2013gwa,Dierigl:2014xta} and Section~\ref{1wallTQFT}. Absent strong symmetry constraints, it is often difficult to uniquely determine the worldvolume TQFT, due to the strong coupling nature of the dynamics  \cite{Cordova:2019jnf,Cordova:2019uob}.
 
 The goal of this paper is to investigate the structure of general $k$-walls\footnote{A $k$-wall connects vacua $k$ units apart, see (\ref{vacua}). Ref.~\cite{Anber:2015kea} considered in some detail only  $k=1$  DWs.} in SYM on $\R^3 \times \S^1$ at small $L N \Lambda$ using semiclassical tools. Our hope is that the results obtained in the calculable regime, generalizing \cite{Anber:2015kea} to arbitrary $N$ and $k$ will help elucidate various still ill understood properties of the domain walls, of their junctions, and of confining strings.

\subsection{Summary of Results}
 
\begin{enumerate}
\item We numerically study the classical  $k$-wall solutions for $2\le N \le 9$. We find that $k$-wall solutions are smooth, with the variations of the fields  within the validity of the effective theory at $\Lambda N L \ll 1$. Details of our numerical studies of DWs are given in Section~\ref{dwnumerics}. Further, we find that, generally, their worldvolumes carry both electric and magnetic fields, whose profiles  we determine.\footnote{To avoid confusion, we stress that the total  magnetic flux carried along a DW is zero.}  The occurrence of  magnetic fields on the DW worldvolume is due to the  nature of magnetic bions---the nonperturbative objects responsible for confinement and the expulsion of electric flux from the vacuum. In effect, magnetic bions create a nonlinear coupling between electric and magnetic fields. This coupling is absent only for an $SU(2)$ gauge group. 

We find that, for $SU(N)$ with $N>2$, magnetic fields are absent only on a finite number of $k=\frac{N}{2}$ BPS walls. We argue that six of these, if $N$ is divisible by $4$ (two solutions, for $N$ even but not divisible by $4$)  carry no magnetic fields, and determine the electric fluxes they can carry. Furthermore, see Appendix \ref{magnetless}, these ``magnetless" solutions can always be expressed in terms of one function, essentially the analytic $SU(2)$ DW solution.

\item Focusing on   the lowest-tension BPS  DWs,  we find numerical confirmation, see Section~\ref{dwnumerics}, of the known result \cite{Hori:2000ck,Acharya:2001dz} that there are  $ \left( \begin{array}{c}N\cr k\end{array} \right)$ BPS  walls between SYM vacua $k$ units apart.\footnote{There is a vast literature on various aspects of DWs in SYM;  an incomplete list is \cite{Ritz:2002fm,Ritz:2004mp,Ritz:2006zz,Dierigl:2014xta,Argurio:2018uup,Draper:2018mpj,Ritz:2018mce,Bashmakov:2018ghn}.} Our new result  is a determination of the electric fluxes carried by the different BPS $k$-walls. The $\left( \begin{array}{c}N\cr k\end{array} \right)$ BPS $k$-walls  carry  Cartan subalgebra electric fluxes whose values fall in one of two groups:\footnote{See the main text for a  detailed explanation. Here we note that $\pmb\rho$ is the Weyl vector and $\pmb{w}_j$, $j=1, ...,N-1$, are the fundamental weights of $SU(N)$.
}  
  \begin{eqnarray}
 2 \pi (\pmb{w}_{i_1} + \ldots + \pmb{w}_{i_k} - \frac{k}{N} \pmb\rho )&,&~~ {\rm there\; are } \;{  \left( \begin{array}{c}N-1\cr k\end{array} \right)} \; {\rm such \; walls },\label{fluxes1}\\
  2 \pi (\pmb{w}_{j_1} + \ldots + \pmb{w}_{j_{k-1}} - \frac{k}{N}\pmb \rho)&,&~~ {\rm there\; are } \;{ \left( \begin{array}{c}N-1\cr k-1\end{array} \right)} \; {\rm such \; walls }.\label{fluxes2}~  \end{eqnarray} 
Here the numbers  $(i_1, ..., i_k)$ are to be taken all different,  ranging from $1$ to $N-1$; likewise all $(j_1, ..., j_{k-1})$ are different.\footnote{For $k=1$, the  set  (\ref{fluxes2})  consists of a single wall carrying flux $-   \frac{2 \pi \pmb \rho}{N}$, hence the number of $k=1$ walls is $N$.}  The above spectrum of BPS $k$-wall fluxes is invariant under $k \rightarrow N-k$ up to reversal of the overall sign of the electric flux (a parity transformation, as in \cite{Bashmakov:2018ghn}).

\item We  use the  results (\ref{fluxes1}, \ref{fluxes2}) for the BPS $k$-wall fluxes to give a  microscopic picture of  the deconfinement of quarks  on DWs. General discrete anomaly inflow arguments and the properties of the worldvolume TQFT  lead  to the conclusion that   the 1-form $\Z_N^{(1)}$ center symmetry is broken on the DW worldvolume, hence quarks should be deconfined there.\footnote{Another property of Wilson loops on DWs with an $\R^3$ worldvolume, i.e. in the $\R^4$ theory, namely their nontrivial braiding (discussed e.g. in \cite{Gaiotto:2013gwa,Hsin:2018vcg}), cannot be seen in the small-$L$ setup of this paper, as the quarks' worldlines would have to cross.} 
Our microscopic picture explaining the phenomenon  generalizes the  one advocated in  \cite{Anber:2015kea}  to general $k$ walls and number of colors $N$. In particular, see Section~\ref{dwdeconfinement}, we argue that our result for the electric fluxes (\ref{fluxes1}, \ref{fluxes2})  implies that  quarks in any nonzero $N$-ality representation are deconfined on $k$-walls for any $k$.
 
\item Finally, we confirm, via numerical minimization, the expectation\footnote{Stemming from the results of \cite{Anber:2015kea}, as explained in EP's  talk at {\it Continuous Advances in QCD-2016}, see \href{https://conservancy.umn.edu/handle/11299/180346}{https://conservancy.umn.edu/handle/11299/180346}.} that our picture of the confining string  implies that heavy static baryons\footnote{We  have in mind  objects composed of $N$ heavy (of mass $M \gg \Lambda$) quarks, approaching the static spectator limit (the color flux picture of Fig.~\ref{fig:03} may be of relevance for excited dynamical baryons). $\Delta$ vs. Y baryons in lattice QCD, also including  light dynamical quarks, are discussed  in \cite{Bali:2003wj}. } in SYM, for $SU(3)$ gauge group,  have color flux in the ``$\Delta$-like" configurations  shown on Figure \ref{fig:03}. The baryon color fluxes are comprised of the three distinct $k=1$ DWs (notice that, in contrast, dYM baryons appear in ``Y" configurations).  The numerical procedure used to study these configurations is described in Section~\ref{dwbaryons}.

The fact that baryons of the $\Delta$ shape for $SU(3)$ (and its polygon generalization for $N>3$) exist in SYM in the present $\R^3\times \S^1$  framework is in marked contrast with another calculable theory of confinement, Seiberg-Witten theory on $\R^4$ \cite{Seiberg:1994rs}, where only linear baryons exist \cite{Douglas:1995nw}. The role of the unbroken zero-form center symmetry and the associated cyclic subgroup of the  Weyl group is crucial for this difference, as discussed in   \cite{Poppitz:2017ivi}.
 \end{enumerate}

\begin{figure}[t]
\begin{subfigure}[t]{.47  \textwidth}
  \includegraphics[width= 1 \textwidth]{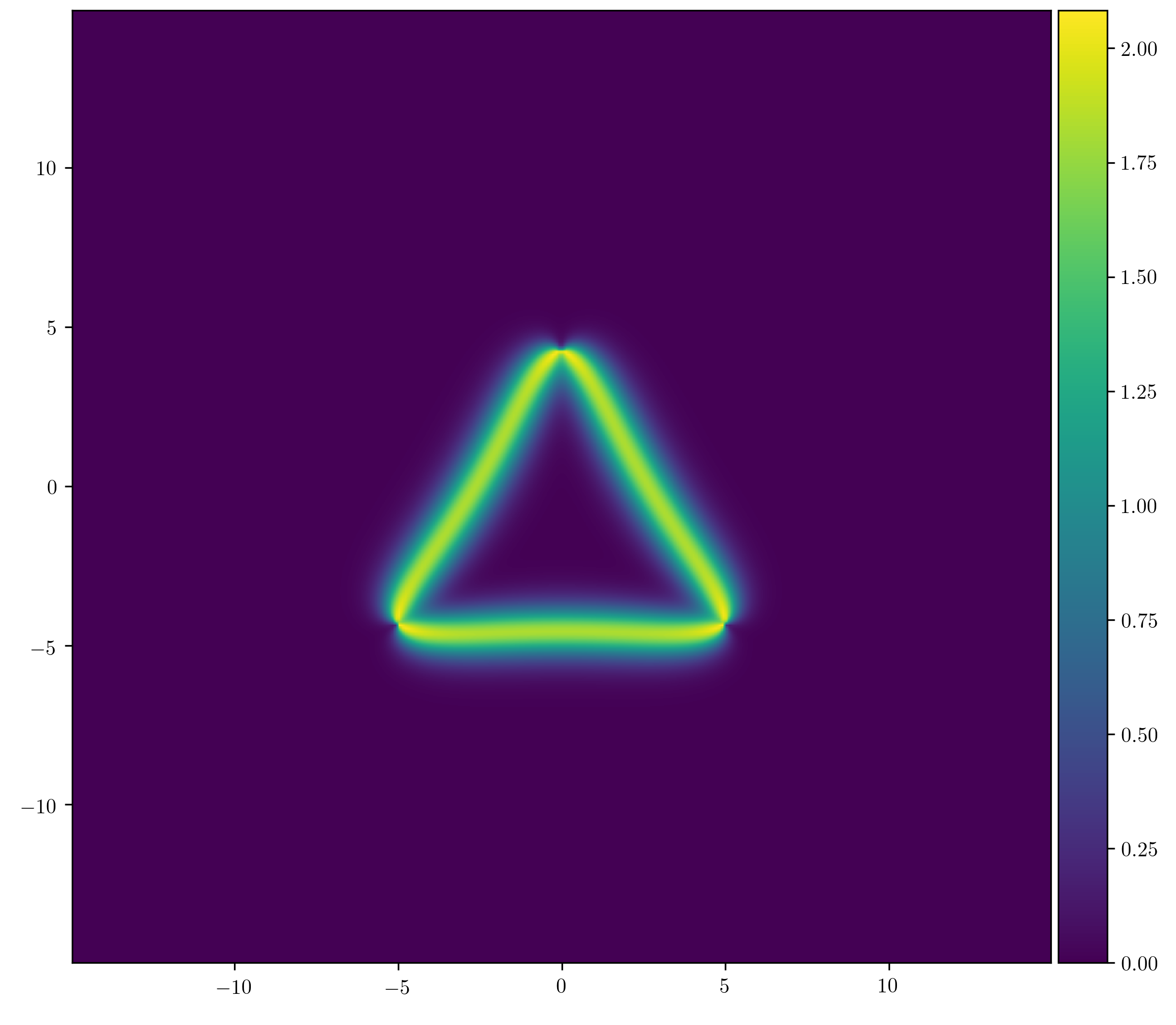}
  \caption{The energy density in a $\Delta$-shape  baryon in $SU(3)$ SYM theory. The three fundamental quarks of weights $\pmb\nu_1, \pmb\nu_2, \pmb\nu_3$ are connected by the three $k=1$ BPS DWs. The vacuum outside of the baryon is $k=0$ and the $k=1$ vacuum is inside.  $N$-sided polygon  configurations occur for $SU(N>3)$, in contrast with the static linear baryons in Seiberg-Witten (SW) theory. This difference arises because there are $N$ magnetic bions in SYM vs.  $N-1$ monopole/dyons in SW.}
  \label{fig:03}
\end{subfigure} 
\qquad
\begin{subfigure}[t]{.47 \textwidth}
  \includegraphics[width= .98\textwidth]{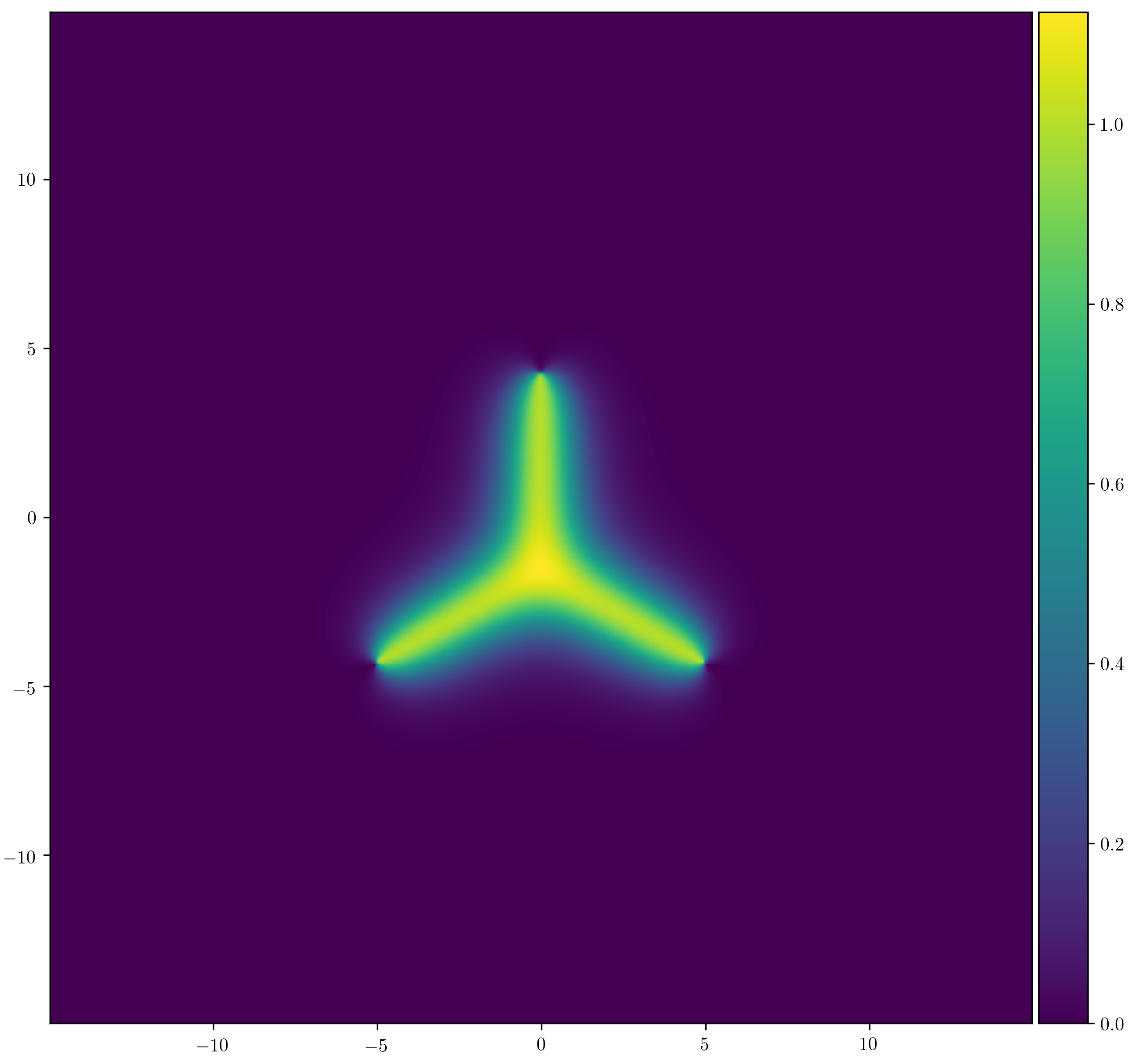}
  \caption{A  Y-shape baryon obtains, instead, in  $SU(3)$ deformed Yang-Mills theory at $\theta =0$. The fundamental quarks are now connected not by DWs ($\theta =0$ dYM has a unique vacuum) but by confining flux tubes whose total flux adds to zero, allowing for the ``baryon vertex" in the middle. At $\theta = \pi$, the corresponding picture is the same as in SYM---a $\Delta$ shape composed of DWs, with one of the two vacua inside the baryon  and the other outside.}
  \label{fig:04}
\end{subfigure}
\caption{Difference between baryons in SYM and dYM on $\R^3 \times \S^1$; see Section \ref{dwbaryons}.}
\end{figure}

	\subsection{Future Work}
	
Continuing this study, it would be of  interest to determine the $N$-ality dependence of the string tensions of the  ``double-string" confining strings. This presents a somewhat challenging numerical problem  left for future work. The $N$-ality dependence of confining string tensions is a probe of the confinement mechanism \cite{Greensite:2011zz} and could be compared with  studies of confining string tensions in both analytic and lattice setups \cite{Bali:2000un,Bringoltz:2008nd,Armoni:2011dw,Greensite:2011gg,Lucini:2012gg,Poppitz:2017ivi,Anber:2017tug}. 
Further, our observation of the $\Delta$ (as opposed to  Y) shape of excited heavy baryons in SYM, should it extend to the theory on $\R^4$, could be viewed as a prediction for the nature of  SYM heavy baryons and could be checked in future lattice simulations of SYM, as in  \cite{Bergner:2015adz,Ali:2018dnd,Ali:2019agk}.

We also note that the nature of DWs, of deconfinement on their worldvolume, and of confining strings has not been addressed in  detail in   other  classes of theories to which similar considerations apply, such as nonsupersymmetric dYM and QCD(adj) on $\R^3 \times \S^1$. It would be interesting to extend the present study to those theories as well.
 
One can also use the walls with fluxes (\ref{fluxes1}, \ref{fluxes2}) to construct  $1/4$-BPS junctions. The simplest example would be a ``star"-like $\Z_N$ symmetric junction of the $N$ different $k=1$ walls; the junction  carries no electric charge, since the fluxes of the walls add up to zero.
There have been discussions about  the possible existence of massless modes on these junctions (in $\R^4$, they have a two-dimensional worldvolume) required by consistency with the conjectured TQFTs on the DW worldvolume \cite{Ritz:2006zz,Gaiotto:2013gwa}. To the best of our knowledge this question is not settled conclusively.  These  DW junctions would be straightforward to produce numerically and the question of existence of massless modes on the junction could be addressed. 

The most interesting question is about the fate of the DWs and quark deconfinement as the size of the circle is increased. It is believed that in SYM this continuation is smooth, at least in the bulk. 
It would be especially interesting to obtain a better physical understanding of the braiding of nonzero $N$-ality Wilson loops predicted by the worldvolume TQFT for $k$-walls on $\R^4$, a phenomenon that our study of DWs with two dimensional worldvolume does not shed light upon.

\subsection{Organization of This Paper}

This paper is roughly divided in two parts. The first half of the paper (Sections \ref{review} and \ref{dwdeconfinement}) uses the conclusions supported by numerics to discuss deconfinement and anomaly inflow. The second half (Sections \ref{dwnumerics} and \ref{dwbaryons}) presents details of the numerical methods used and discusses many of the results. 

We begin in Section~\ref{review}, with a review of SYM in the semiclassical limit, its relevant symmetries, and vacua. In Section~\ref{dwdeconfinement},   borrowing the results for DW fluxes---found in the second half of the paper and already stated above---we discuss the action of symmetries on these fluxes. We use the results for $k$-wall fluxes to arrive at the conclusion that all $N$-ality quarks are deconfined on $k$-walls. An alternative description of the worldvolume of DWs and anomaly inflow is via a TQFT, which we write explicitly for the case of $k=1$ BPS walls.

The second part of the paper is devoted to the details of the   numerical methods used, and presents the numerical and analytical results regarding the DW properties. In Section~\ref{dwnumerics}, we discuss the methods used, present a few of the profiles of $k$-walls and point out examples of profiles exhibiting various interesting properties. In particular, we summarize the results regarding the ``magnetless" DWs (the ones that do not  have magnetic fields), whose derivation is given in Appendix \ref{magnetless}.

 In Section~\ref{dwbaryons}, we describe the numerical procedures used to calculate the energy of static fundamental quarks as a function of their distance along a DW, giving an explicit numerical confirmation of the picture of deconfinement advocated earlier and shown on Fig.~\ref{fig:02}. We also describe how baryon configurations, recall Fig.~\ref{fig:03}, are studied. In Sections~\ref{dwnumerics}, \ref{dwbaryons}, we describe the methods used in some detail, because they could be useful in future studies of baryons, deconfinement, and $k$-wall junctions in a variety of semiclassical theories, e.g. SYM, dYM, QCD(adj) on $\R^3 \times \S^1$.  

\section{Brief Review of SYM on $\R^3 \times \S^1$}
\label{review}
   
We shall not dwell into the details of the microscopic dynamics leading to the long-distance theory\footnote{In historical order, see \cite{Seiberg:1996nz, Aharony:1997bx} and the instanton calculation of \cite{Davies:2000nw,Davies:1999uw} (completed recently in \cite{Poppitz:2012sw,Anber:2014lba}).} described below. Our starting point here is the result that 
the infrared dynamics of SYM on $\R^3 \times \S^1$ is described by a theory of chiral superfields. Their  lowest components are 
the Cartan-subalgebra valued  bosonic fields, the dual photons and holonomy scalars combined into complex scalar fields:
\begin{eqnarray}
\label{fields1}
x^a = \phi^a + i \sigma^a~, ~ a= 1, \ldots N-1, ~~ {\rm with } ~\; \sigma^a  \simeq  \sigma^a + 2 \pi w_k^a~, ~ k = 1, \ldots N-1~.  
\end{eqnarray}
Here $\phi^a$ are real scalar fields describing the deviation of the $\S^1$ holonomy eigenvalues from its center-symmetric value and $\sigma^a$ are the duals of the Cartan-subalgebra photons; both fields are taken to be dimensionless.  
The more precise relation to 4d fields is: 
\begin{eqnarray} \label{fieldrelation}
{g^2  \over 4\pi L}\; \partial_\mu \phi^a &=& F_{\mu 4}^a,~~  \mu,\nu = 0,1,2 ,\nonumber \\
{g^2 \over 4 \pi L}\; \epsilon_{\mu\nu\lambda} \partial^\lambda \sigma^a &=&  F_{\mu\nu}^a.
\end{eqnarray}
Here $F_{\mu 4}^a$  denotes the mixed $\R^3$-$\S^1$ component of the Cartan field strength tensor of the 4d theory and $F_{\mu\nu}^a$ is the field strength along $\R^3$, all taken independent of the $\S^1$-coordinate $x^4$, and $g^2$ is the 4d SYM gauge coupling at a scale of order $1\over L$  (see \cite{Anber:2014lba} for more detail, including renormalization).
Thus, Eqn.~(\ref{fieldrelation}) implies that spatial derivatives of $\phi^a$ are 2d duals of the 4d magnetic field's components along  $\R^2$;  likewise, spatial derivatives of $\sigma^a$ are 2d duals of the 4d electric field's  components along  $\R^2$.

As indicated in (\ref{fields1}), the target space of the  dual photons fields is the unit cell of the $SU(N)$ weight lattice spanned by $\pmb w_k = (w_k^1, ..., w_k^{N-1})$, $k=1\ldots N-1$, the fundamental weights.\footnote{Vectors in the Cartan subalgebra will be denoted by bold face: $\pmb \sigma = (\sigma^1, ..., \sigma^{N-1})$, $\pmb \phi = (\phi^1, ..., \phi^{N-1})$ and the complex field $\pmb x = (x^1, ..., x^{N-1})$ and similar for their complex conjugates $\bar{\pmb x}$. The dot product used throughout is the usual Euclidean one. } A simple way to understand the $\sigma$-field periodicity is that it  allows for non vanishing monodromies corresponding to the insertion of probe electric charges (quarks) of any nonzero $N$-ality, as appropriate in an $SU(N)$ theory.  

At small $L N \Lambda \ll 1$, the bosonic part of the long distance theory is described by the weakly-coupled $\R^{3}$ lagrangian:
\begin{equation}
\label{lagrangian1}
L = M \; \partial_\mu { x}^a g_{ab} \; \partial^\mu { \bar{x}}^b - M\; {  m^2  \over 4}\; {\partial W({\pmb x}) \over \partial x^a} \; g^{ab}\; {\partial  \bar W(\bar{\pmb x}) \over \partial \bar{x}^{b}}  ~.
\end{equation}
The spacetime metric is $(+,-,-)$ and $W(\pmb{x})$ is the holomorphic superpotential:
\begin{equation}
\label{lagrangian2}
W({\pmb x}) = \sum_{a=1}^N e^{\pmb\alpha_a \cdot {\pmb x}}~.
\end{equation}
Here $\pmb\alpha_{1}, \ldots, \pmb\alpha_{N-1}$ are the simple roots and $\pmb\alpha_N = - \sum\limits_{a=1}^{N-1} \pmb\alpha_a$ is the affine, or lowest, root of the $SU(N)$ algebra.\footnote{Roots are normalized to have length $2$; roots and coroots are identified, and $\pmb \alpha_a \cdot \pmb{w}_b = \delta_{ab}$, $a,b=1, \ldots N-1$.} 
The K\" ahler metric  appearing in (\ref{lagrangian1}) is 
\begin{equation}
\label{kahler}
g_{ab} = \delta_{ab} + \ldots~,
\end{equation} and $g^{ab} \approx \delta^{ab}$ is its inverse. We stress that the above minimal form of the K\" ahler metric is not an assumption:  the form of $g_{ab}$ is justified in the semiclassical $L N \Lambda \ll 1$ limit. The dots in (\ref{kahler}) indicate   corrections  computed in \cite{Poppitz:2012sw,Anber:2014sda,Anber:2014lba}. We shall ignore them, as they are negligible in the semiclassical limit (when taken at finite $N$ \cite{Cherman:2016jtu}).

The scales appearing in the long-distance theory (\ref{lagrangian1}) are determined by the dynamics of the underlying 4d SYM theory. We do not need the precise values but only note that $M \sim {g^2 \over L}$, where $g^2$ is the SYM coupling at the scale $L$,  stressing again that $g$ is small in the semiclassical $L N \Lambda \ll 1$ limit. The scale $m \sim M e^{- {8 \pi^2 \over N g^2}}$ is a nonpertubative scale generated by various semiclassical monopole-instantons. As is clear from (\ref{lagrangian1}, \ref{lagrangian2}), $m$ sets the mass scale of the $\phi$ and $\sigma$ fields  and their superpartners. In what follows, we shall often call the scale $m$ the ``dual photon mass." We should, however, keep in mind that $m$ is really the mass of the heaviest of the $N-1$ dual photons, whose mass spectrum is given by $m_k \sim m \sin^2 {\pi k \over N}$, $k=1,...,N-1$.\footnote{See \cite{Cherman:2016jtu} for a discussion of the large-$N$ limit, \cite{Anber:2017ezt} for the study of bound states of dual photons, the 3d remnant of 4d glueballs
\cite{Aitken:2017ayq}  bound by doubly-exponential nonperturbative effects.} We find that the widths of DWs are generally determined by the lightest dual photon mass (with the notable exception of the magnetless walls, see Section \ref{magnetlesssolutions}).
 
{\flushleft{Of}} special interest to us here are the discrete chiral and center global symmetries of SYM:
\begin{enumerate}
\item A 0-form chiral symmetry, $\Z_{2 N}^{(0)}$. This is the usual discrete $R$-symmetry of SYM that  acts on the fermionic superpartners of (\ref{fields1}) (and on the superpotential),  by a phase rotation. However, of most relevance to us is that  it also shifts the dual photons:
\begin{eqnarray}\label{chiral}
\Z_{2 N}^{(0)}: ~\pmb \sigma &\rightarrow& \pmb \sigma + {2 \pi \pmb \rho \over N}~,  \\
e^{\pmb\alpha_a \cdot \pmb{x}} &\rightarrow& e^{i {2\pi \over N}} e^{\pmb\alpha_a \cdot \pmb{x}} ~,~ a = 1, \ldots, N,\nonumber\end{eqnarray}
where we also show the action of  $\Z_{2 N}^{(0)}$ on the terms appearing in the superpotential (\ref{lagrangian2}). We denote the Weyl vector by $\pmb\rho = \pmb{w}_1 + \ldots +\pmb{w}_{N-1}$.
\item A ``0-form" center symmetry
 $\Z_N^{(1),\S^1_L}$. The notation is chosen to emphasize that this is the dimensional reduction of the    $\S^1_L$-component of the 4d 1-form center symmetry. As explained in detail in 
 \cite{Anber:2015wha,Cherman:2016jtu,Aitken:2017ayq}, the  0-form center symmetry  acts on the dual photons and their superpartners:
 \begin{eqnarray} 
 \Z_N^{(1),\S^1_L}: ~ \pmb{x} &\rightarrow& {\cal P} \pmb{x},\label{zeroform0} \\
 e^{\pmb\alpha_a \cdot \pmb{x}} &\rightarrow&  e^{\pmb\alpha_{a+1 ({\rm mod} N)} \cdot \pmb{x}}~ \label{zeroform1},
 \end{eqnarray}
 and is an exact unbroken global symmetry of the theory (also respected by the K\" ahler potential terms omitted from (\ref{kahler})).
 
In a basis independent way, the operation denoted by $\cal{P}$  in (\ref{zeroform0})  is the product of Weyl reflections w.r.t. all simple roots \cite{Anber:2015wha},   while in the standard $N$-dimensional basis for the weight vectors it is a $\Z_N$ cyclic permutation of the vector's components \cite{Cherman:2016jtu,Aitken:2017ayq}. This clockwise  action is evident from the way  $\Z_N^{(1),\S^1_L}$ acts on the  $e^{\pmb\alpha_a \cdot \pmb{x}}$ terms in the superpotential, eqn. (\ref{zeroform1}). This symmetry shall be important in what follows; its action on electric fluxes of DWs is given below in (\ref{zeroform2}).
 \item A 1-form center symmetry $\Z_N^{(1),\R^3}$, acting on Wilson  line operators in $\R^3$  by multiplication by appropriate $\Z_N$ phases.
 \end{enumerate}
 The 0-form and 1-form center symmetries discussed above are parts of the  $\Z_{N}^{(1), \R^4}$ 1-form center of the $\R^4$ theory. 
As recently realized \cite{Gaiotto:2014kfa}, SYM on $\R^4$ has a mixed 't Hooft anomaly between the 0-form  $\Z_{2N}^{(0)}$ chiral symmetry and the $\Z_{N}^{(1), \R^4}$ 1-form center symmetry. This 't Hooft anomaly persists upon an $\R^3 \times \S^1$ compactification and anomaly matching is saturated by the spontaneous breaking of the discrete chiral symmetry,  $\Z_{2N}^{(0)} \rightarrow \Z_2^{(0)}$, as in (\ref{vacua}) below. Anomaly inflow on the resulting DWs implies that the DW worldvolume is nontrivial, supporting the phenomena of quark deconfinement (and braiding of Wilson lines, for DWs with three-dimensional worldvolume). Aspects of this inflow has been  studied in many works \cite{
Gaiotto:2017yup,Gaiotto:2017tne,Argurio:2018uup,Draper:2018mpj,Ritz:2018mce,Bashmakov:2018ghn,Anber:2018jdf,Anber:2018xek}. Here, we continue  its study via semiclassical tools.

It is well known that SYM on $\R^3 \times \S^1$,  (\ref{lagrangian1},\ref{lagrangian2}), has $N$ vacua, determined by solving for the stationary points of the superpotential $W(\pmb x)$. These are labelled by $k= 0, \ldots, N-1$. All vacua have $\langle\pmb\phi\rangle = 0$. The dual photon field $\pmb\sigma$ has nontrivial expectation value:
\begin{eqnarray}
\label{vacua} 
\langle \pmb{\sigma} \rangle_k &=& { 2 \pi k \over N} \pmb\rho ~, \\
\langle X_a \rangle_k  &=& \langle e^{\pmb\alpha_a \cdot \pmb{x}} \rangle_k  =  e^{i {2 \pi k \over N}}, a = 1,\ldots N,\nonumber  \\
\langle W \rangle_k &\equiv& W_k  =  N e^{i {2 \pi k \over N}}~. \nonumber
\end{eqnarray}
We introduced the notation $X_a \equiv e^{\pmb\alpha_a \cdot \pmb{x}}$ (such that    $X_1 X_2... X_N = 1$, to be used in some discussions below, following \cite{Hori:2000ck,Ritz:2002fm}), a set of $N$ fields which are single-valued on the Cartan torus and do not allow for describing nonvanishing monodromies. We also denoted the   expectation value of the superpotential in the $k$-th ground state by $W_k$. The $N$ vacua (\ref{vacua}) are interchanged by the action of the spontaneously broken $\Z_{2 N}^{(0)} \rightarrow \Z_2^{(0)}$ symmetry (\ref{chiral}), while the $\Z_N^{(1),\S^1_L}$ symmetry is unbroken.\footnote{This follows from ${\cal{P}} {2 \pi \pmb\rho \over N} = {2\pi \pmb\rho \over N} - 2 \pi \pmb{w}_1$, see (\ref{zeroform2}). In words,  the action of the 0-form center $\Z_N^{(1),\S^1_L}$ on the vacua (\ref{vacua}) is a weight-lattice shift of $\langle\pmb\sigma\rangle_k$, which is an identification, as per (\ref{fields1}); see also Figure~\ref{fig:su3}.} The 1-form $\Z_N^{(1),\R^3}$ symmetry is also unbroken in the bulk of SYM, corresponding to the confinement of quarks.

It may be helpful  to visualize the fundamental domain of $\pmb\sigma$, the action of the 0-form discrete chiral and center symmetries, and the vacuum structure. We show this in the simple case of $SU(3)$ SYM on Figure~\ref{fig:su3}.

  \begin{figure}[t]
   \centering
  \includegraphics[width=.60\linewidth]{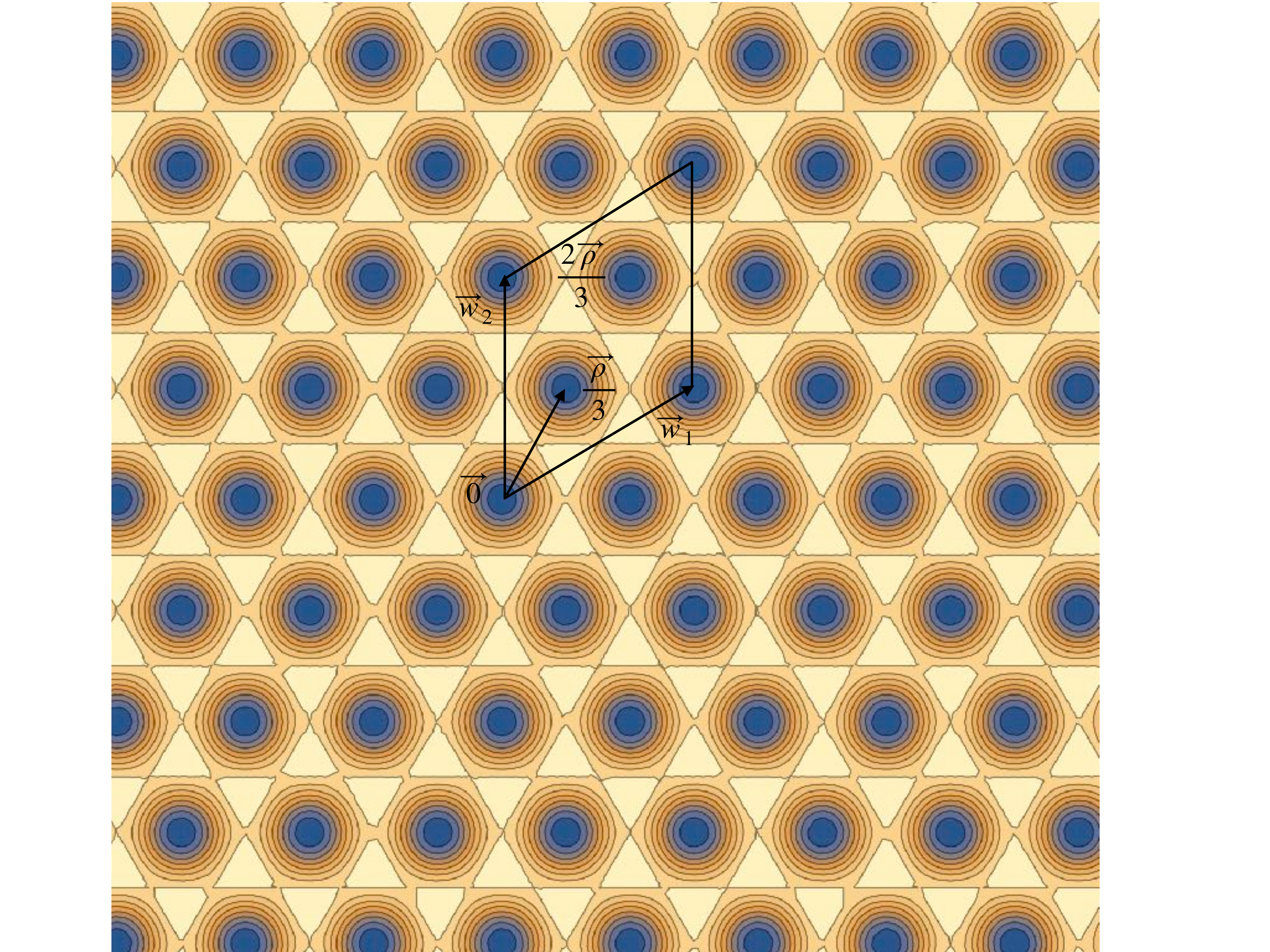}
  \caption{The  $\pmb\sigma \over 2 \pi$-plane in $SU(3)$ SYM.  The fundamental domain of the dual photons (\ref{fields1})  is the unit cell of the weight lattice spanned by $\pmb{w}_1$ and $\pmb{w}_2$. The contours on the plot are equipotential surfaces, see (\ref{lagrangian1}), as a function of $\pmb\sigma\over 2 \pi$, with $\pmb\phi$ set to $0$.
   The three vacua ${\pmb\sigma \over 2 \pi}=0$, $ {  \pmb\rho\over 3}$, ${2 \pmb\rho\over 3}$ in the fundamental domain are shown in dark dots. The $\Z_3$ discrete chiral symmetry acts on $\pmb\sigma \over 2 \pi$ by $\pmb\rho\over 3$ shifts. The 0-form $\Z_3$-center symmetry acts by $120$-degree rotations around the origin. That this symmetry is unbroken in the three vacua follows from applying a $\Z_3$ rotation and a subsequent weight lattice shift. The breaking of the $\Z_3$ 0-form center on the DWs follows from the fact that this symmetry relates different DW solutions: the three $k=1$ BPS DWs have electric fluxes proportional to $\pmb\rho/3$, $\pmb{w}_1 -\pmb\rho/3$ and $\pmb{w}_2 -\pmb\rho/3$, which can be represented by the corresponding vectors on the plot, clearly related by $Z_3$ rotations. }  \label{fig:su3}
\end{figure}

As usual, there are domain walls (DW) connecting the various discrete vacua. A DW is a static configuration on $\R^{3}$ connecting two vacua. While a more appropriate name would be a ``domain line" (as their worldvolume is two-dimensional), we continue to call them DW. The tension of the DW is its energy per unit length. 
A DW connecting vacua $k$ units apart, i.e. stretching between $W_p$ and $W_{p+k ({\rm mod} N)}$, is called a ``$k$-wall."  The physics of the DWs in SYM theory is quite rich and has been the subject of many investigations over the past 20 years, for example \cite{Hori:2000ck,Acharya:2001dz,Ritz:2002fm,Ritz:2004mp,Ritz:2006zz,Dierigl:2014xta,Argurio:2018uup,Draper:2018mpj,Ritz:2018mce,Bashmakov:2018ghn}.

\section{$k$-Wall Fluxes and Deconfinement of Quarks on DWs}
\label{dwdeconfinement}

We begin with some remarks regarding confinement in SYM on $\R^3 \times \S^1$. Most importantly, the theory  abelianizes in the semiclassical regime. Consider then the Wilson loop operator, in a representation $\cal{R}$, taken around some loop $C \in \R^3$. At scales $\gg L$, abelianization reduces this operator  to  the unbroken Cartan-subalgebra Wilson loop. The expectation value of its trace can thus be expressed as a sum over the weights  $\pmb\lambda_b$ of the representation $\cal{R}$:
\begin{eqnarray}\label{wilsonwall}
\langle W_{\cal{R}}(C) \rangle =  \langle {\rm tr}_{\cal{R}} \; e^{i \oint_C A} \rangle\big\vert_{N L \Lambda \ll 1} \rightarrow \sum\limits_{b=1}^{{\rm dim}({\cal R})} \langle e^{i \; \pmb\lambda_b \cdot \oint \pmb A}\rangle~,
\end{eqnarray}
where each term corresponds to the insertion of a quark with worldline along $C$ and electric charge given by  one of the weights  of $\cal{R}$. The expectation value of the Wilson loop for each weight, $\langle e^{i \; \pmb\lambda_b \cdot \oint \pmb A}\rangle$, is computed semiclassically. One begins by realizing that the insertion of this operator imposes a monodromy of the dual photons $\pmb \sigma$ around the loop $C$ and then solves for the field configuration of minimal action that has the right monodromy. In the Polyakov model or in deformed Yang-Mills theory on $\R^3 \times \S^1$ this can be done analytically (for $SU(2)$ and $SU(3)$, e.g.~\cite{Poppitz:2017ivi}) in the limit of an infinite loop   spanning the entire $xy$ plane in $\R^3$. However, the semiclassical configurations extremizing the path integral expression for  $\langle e^{i \; \pmb\lambda_b \cdot \oint \pmb A}\rangle$ in SYM are not easy to describe analytically, since they correspond to the ``double string" picture of Fig.~\ref{fig:01}. In Section~\ref{dwbaryons}, we explain in detail how to formulate the problem and how to obtain the picture on that figure using numerical methods.

We now imagine inserting the Wilson loop $W_{\cal{R}}(C)$ in the  $\R^2_{wall}$ worldvolume of a $k$-wall, i.e. we take 
 $C \in \R^2_{wall}$.  To show that the Wilson loop (\ref{wilsonwall}) exhibits perimeter law in the limit of large contour $C$, with the number of colors $N$ and dim($\cal{R}$) kept fixed as $C$ becomes large, it suffices  to show that at least one of the weights $\pmb\lambda_b$ of the representation $\cal{R}$ is not confined on any $k$-wall. This is because the sum on the r.h.s. of (\ref{wilsonwall}) will contain terms exhibiting perimeter law, $\sim e^{-   P(C)}$, due to the deconfined components (weights) of quarks,  as well as possible terms exhibiting area law,  $\sim e^{-   A(C)}$, due to the confined weights. In the limit of large $C$, the perimeter law terms dominate the sum, leading to the conclusion that the Wilson loop (\ref{wilsonwall}) exhibits perimeter law on the wall.
These remarks will be useful in Section~\ref{deckwall}.

\subsection{Quark Deconfinement on DWs}
 \begin{figure}[h]
   \centering
  \includegraphics[width=.60\linewidth]{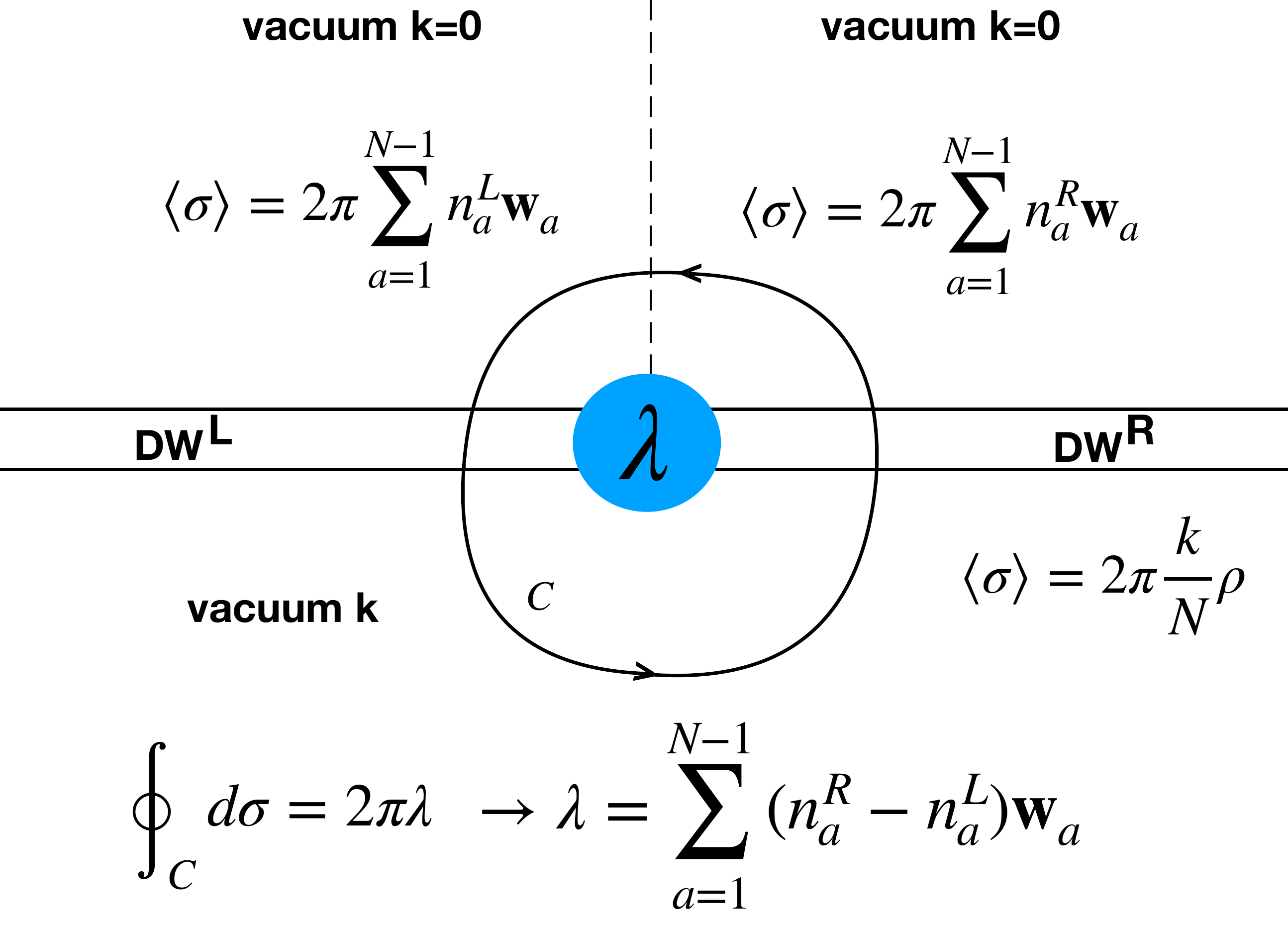}
  \caption{The weight  of a quark 'suspended' on a  wall between vacua $k$ units apart equals the difference between the electric fluxes  of the DWs to the left and the right. Deconfinement of a quark anti-quark pair on the wall further requires that the tensions of the two walls be equal.}  \label{fig:ABC}
\end{figure}
The basic mechanism of quark deconfinement on DWs is illustrated on Fig.~\ref{fig:ABC}. The DWs to the left and right of the quark of weight $\pmb\lambda$ are denoted by $DW^{L, R}$. These DWs connect the vacuum $k=0$, or the ones equivalent to it by weight-lattice shifts (in the half-plane above the walls)  to the $k$-th vacuum (the lower half-plane). The charge, or weight $\pmb\lambda$, of the quark under the Cartan subalgebra $U(1)$'s is equal to the monodromy of $\pmb\sigma$ around the contour $C$ and is thus   determined by the difference of electric fluxes carried by the two walls, as shown in the equation on the bottom of the figure and in (\ref{lambdaonwall}) below. Further, when the tensions of $DW^L$ and $DW^R$ are equal, the quark of weight $+\pmb\lambda$ and an antiquark of weight $-\pmb\lambda$, located some distance apart along the wall (the antiquark is not shown here, but recall Fig.~\ref{fig:02}, or Fig.~\ref{fig:Deconfinement Reverse Electric}) can independently move along the wall, changing their relative distance without any cost of energy (see Fig.~\ref{fig:deconfinement_energy} for a plot of the energy as a function of distance). Equal tension of the two  $k$-walls is guaranteed if they are both BPS, i.e. carry one of the fluxes listed in (\ref{fluxes1}, \ref{fluxes2}). The tensions are also equal if the two walls are non-BPS, but are related by the $\Z_{N}^{(1), \S^1_L }$ 0-form center symmetry. The dashed vertical line on Fig.~\ref{fig:ABC} shows where a weight-lattice discontinuity of $\pmb\sigma$ is imposed (its precise location or shape are irrelevant, but in the  simulations that produced Fig.~\ref{fig:02}, a straight line discontinuity  was used, see Section~\ref{dwbaryons} for details). 

In summary, a quark of weight $\pmb\lambda$ can have its flux split between two $k$-walls that carry fluxes $2 \pi (\sum\limits_{a=1}^{N-1}{n}^L_a   \pmb w_a - {k\over N} \pmb\rho)$ and $2 \pi (\sum\limits_{a=1}^{N-1}{n}^R_a   \pmb w_a - {k\over N} \pmb\rho)$, respectively, provided its weight satisfies 
\begin{equation}
\label{lambdaonwall}
\pmb\lambda = \sum\limits_{a=1}^{N-1}({n}^R_a - n^L_a)   \pmb w_a~.
\end{equation}
This relation only uses charge conservation. In addition, as explained above, deconfinement of a quark and antiquark placed some distance apart on the $k$-wall requires that the tensions of the walls to the left and the right be equal. 

\subsection{The 0-form Center Symmetry and Its Action on DW Fluxes}
In our discussion below, the $\Z_N^{(1), \S^1_L}$ 0-form center symmetry (\ref{zeroform0},\ref{zeroform1}) will play a crucial role.
The $\Z_N^{(1), \S^1_L}$ action on the electric flux of DWs (i.e., on the change of $\pmb\sigma$ across the walls) can be found using the definition of ${\cal P}$ as a product of Weyl reflections wrt all simple roots (recall (\ref{zeroform1})). Taken to act on the difference between the  boundary values of $\pmb\sigma$ on the two sides of a wall, it follows that the $\Z_N^{(1), \S^1_L}$ action is: 
\begin{eqnarray} \label{zeroform2}
{\cal P}^l  \; {k  \over N} \pmb\rho &=& {k\over N} \pmb\rho - k \;\pmb{w}_l~, ~l =1, \ldots N~, \nonumber \\
{\cal P}^l \; \pmb{w}_a &=&  \pmb{w}_{a+l ({\rm mod} N)} -  \pmb{w}_l~,~ a = 1, \ldots N-1, ~ {\rm and \; we \; defined} ~~ \pmb{w}_0 = \pmb{w}_N = 0~. 
\end{eqnarray}
Suppose now the electric flux of some $k$-wall (BPS or not) is given by\footnote{For simplicity, we only consider fluxes $\pmb\Phi$  within the fundamental domain of $\pmb\sigma$, i.e. such that all weight vectors occur with coefficients $0$ or $1$ only. The action of the $\Z_{N}^{(1),\S^1_L}$ center symmetry can be generalized for general fluxes (i.e. non-BPS walls) but we shall not need the corresponding expression here.}
\begin{equation}
    \pmb{\Phi}=\sum_{i=1}^{m}\pmb{w}_{\pi(i)}-\frac{k}{N}\pmb{\rho}
\end{equation}
where $\pi$ is some permutation of the indices $1$ to $N-1$. The total number of weights appearing in the sum is $m$, with $0\leq m\leq N-1$, and we dropped factors of $2\pi$ in the definition of electric fluxes. Then, acting $l$ times with $\mathcal{P}$ on $\pmb{\Phi}$ gives us
\begin{align}
    \mathcal{P}^l \;\pmb{\Phi}
    %\sum_{i=1}^{m}\mathcal{P}^l\left(\pmb{w}_{\sigma(i)}\right)-\mathcal{P}^l\left(\frac{k}{N}\pmb{\rho}\right)\nonumber\\
    %&=\sum_{i=1}^{m}\left(\pmb{w}_{\sigma(i)+l}-\pmb{w}_l\right)-\left(\frac{k}{N}\pmb{\rho}-k\pmb{w}_l\right)\nonumber\\
    %&=\sum_{i=1}^{m}\pmb{w}_{\sigma(i)+l}-m\pmb{w}_l-\frac{k}{N}\pmb{\rho}+k\pmb{w}_l\nonumber\\
    =\sum_{i=1}^{m}\pmb{w}_{\pi(i)+l}+\left(k-m\right)\pmb{w}_l-\frac{k}{N}\pmb{\rho}\label{eqn:action_flux}~.
\end{align} 
Now, we can ask whether $\pmb{w}_l$ appears in the sum ($\sum_{i=1}^{m}$) in equation (\ref{eqn:action_flux}). If it does, then there must be some $i$ between $1$ and $m$ for which $\pi(i)+l({\rm mod}{N})\equiv l$, which implies that $\pi(i)\equiv 0({\rm mod}{N})$. However, we defined the permutation $\pi$ such that this is not possible, and therefore we know that $\pmb{w}_l$ does not appear in the sum. 

Let us now require that $\mathcal{P}^l$ takes solutions with  electric flux  in the fundamental domain (such that all the coefficients on the $\pmb{w}_j$ are either $0$ or $1$) to solutions  whose  electric flux is also in the fundamental domain. This implies  that the coefficient  in front of all $\pmb{w}_j$ appearing in equation (\ref{eqn:action_flux}) must be $0$ or $1$ (a coefficient bigger than $1$ would mean that the solution leaves the fundamental domain). In the sum term, it is obvious that the coefficients are all $0$ or $1$, so the only term that we must consider is the second term, $\left(k-m\right)\pmb{w}_l$. Requiring the result to be in the fundamental domain tells us that $k-m$ must be either $0$ or $1$, and so $m$ must be either $k$ or $k-1$. This means that  we must have started at a solution with either $k$ or $k-1$ of the $\pmb{w}_j$. But it does not matter where we started within the orbit, and therefore any solution in the fundamental domain for a $k$-wall which maps, under the action of $\mathcal{P}$, to other $k$-wall solutions with electric fluxes in the fundamental domain must have either $k-1$ or $k$ of the $\pmb{w}_j$.

 To summarize, in Eqns.~(\ref{zeroform2}, \ref{eqn:action_flux}), we found how the  $\Z_N^{(1), \S^1_L}$ symmetry acts on the DW fluxes. 
This result, along with the conclusion from the previous paragraph,  implies that the $ \left( \begin{array}{c}N \\ k \end{array} \right)$ walls with fluxes (\ref{fluxes1}, \ref{fluxes2}) are the only $k$-walls such that a $\Z_N^{(1), \S^1_L}$  transformation maps walls whose flux is within the fundamental domain (\ref{fields1}) to walls whose flux is also inside the fundamental domain.  The significance of this observation is as follows.  Walls where $\pmb\sigma$ leaves the fundamental domain of the weight lattice are such that some of the complex covering space coordinates $X_1,... X_N$ (introduced in (\ref{vacua}))  necessarily  have nontrivial winding  around the origin of the complex plane as the wall is traversed. As 
 discussed in \cite{Hori:2000ck,Ritz:2002fm},   $k$-walls with winding $X$-space trajectories are
 non-BPS. On the other hand,   the $  \left( \begin{array}{c}N \\ k \end{array} \right)$ BPS $k$-walls are walls with no  winding for any of the $X$'s; below, we call these ``nonwinding." Being a symmetry of the theory, $\Z_N^{(1), \S^1_L}$ maps  BPS walls into BPS walls.  Since the previous argument  shows that  the walls with fluxes (\ref{fluxes1}, \ref{fluxes2}) comprise the only set of only ``nonwinding" walls closed under the 0-form center symmetry, these walls must be all the BPS walls.

\subsection{Deconfinement on $k$-Walls}
\label{deckwall}

Next, we use the result for the fluxes of BPS $k$-walls (\ref{fluxes1},\ref{fluxes2}) to argue that Wilson loops in arbitrary $SU(N)$ representations of nonzero $N$-ality  obey perimeter law on $k$-walls. As discussed in the beginning of this Section, it suffices to show that at least one weight of any nonzero $N$-ality representation is deconfined.
Our knowledge of the BPS $k$-wall fluxes (\ref{fluxes1}, \ref{fluxes2}) is sufficient to argue that for any $\cal{R}$ of nonzero $N$-ality $q$ ($1\le q < N$), there is a deconfined weight. This already follows from the fact\footnote{A proof is in, e.g. Appendix C.1 of \cite{Poppitz:2017ivi}.} that the $q$-th fundamental weight $\pmb{w}_q$ is a weight of any representation $\cal{R}$ of $N$-ality $q$. Quarks of this weight are deconfined, as shown on Fig.~\ref{fig:ABC}, by two BPS $k$-walls whose fluxes  (\ref{fluxes1},\ref{fluxes2}) can be taken as
$\pmb{w}_{i_1}+\pmb{w}_{i_2}+ \ldots +\pmb{w}_{i_{k-1}} - {k \over N} \pmb \rho$ and $\pmb{w}_{i_1}+\pmb{w}_{i_2}+ \ldots +\pmb{w}_{i_{k-1}} + \pmb{w}_q - {k \over N} \pmb \rho$ (where $i_1, ...i_{k-1}$ are all taken  different from $q$).
An especially simple example is the Wilson loop in the fundamental representation. It is easily seen that all its weights are deconfined on BPS $k$-walls. This is because the weights of the fundamental representation are $\pmb\nu_A = \pmb{w}_A - \pmb{w}_{A-1}$ ($A=1, \ldots N$, with $\pmb{w}_0\equiv \pmb{w}_N \equiv 0$). The weight $\pmb\nu_A$  is always equal to the difference between two of the BPS $k$-wall fluxes (\ref{fluxes1},\ref{fluxes2}).\footnote{In fact, for $k>1$, it is easy to see that there is a degeneracy in the choice of $k$-walls on which fundamental quarks are deconfined. The source of this degeneracy is the increasing degeneracy of BPS $k>1$-walls.} 

In conclusion, the study of BPS wall fluxes on $\R^3 \times \S^1$ allows  us to conclude that $SU(N)$ gauge invariant  Wilson loops of nonzero $N$-ality in $\R^3$ exhibit perimeter law on domain walls separating vacua $k$ units apart. Thus, the $\Z_N^{(1), \R^3}$ center symmetry under which these loops are charged is broken on the DW, as required by anomaly inflow arguments. 

\subsection{The ${k=1}$ wall worldvolume TQFT and anomaly inflow}
 \label{1wallTQFT}

The simplest  walls to consider in more detail are the $k=1$ walls. There are exactly $N$  such walls of fluxes $2 \pi (\pmb{w}_a - {1\over N}  \pmb\rho)$, $a=0,..., N-1$, with $\pmb{w}_0 = 0$. These walls fill out an $N$-dimensional orbit of the $0$-form center symmetry. 

A different  way to describe the $N$  $k=1$ walls is in terms of a worldvolume TQFT.\footnote{We ignore the decoupled ``center of mass" degrees of freedom.} It was suggested in  \cite{Anber:2018xek} that the worldvolume theory is the $\Z_N$ 2d topological gauge theory\begin{equation}
\label{1walltqft}
S_{1-wall} =   {N \over 2 \pi} \int\limits_{\R^2} \phi^{(0)} d a^{(1)}~.~
\end{equation}
Here  $\R^2$ denotes  the DW worldvolume, 
$\phi^{(0)}$ is a compact scalar with period $2\pi$ transforming under the $0$-form center symmetry $\Z_N^{(1),\S^1_L}$, $\phi^{(0)} \rightarrow \phi^{(0)} + {2 \pi \over N}$,  and $a^{(1)}$ is a compact gauge field transforming under the $1$-form center $\Z_N^{(1),\R^2}$, $a^{(1)} \rightarrow a^{(1)} + {1 \over N} \epsilon^{(1)}$, where $\epsilon^{(1)}$ is a closed $1$-form with period $2\pi$. 

 The TQFT (\ref{1walltqft}), upon quantization, can be seen to have $N$ ground states related by the broken $0$-form center symmetry $\Z_N^{(1),\S^1_L}$. These ground states can be identified with the $N$ different  BPS $1$-walls found in this paper. That the description using (\ref{1walltqft}) is correct is also supported by the fact that (\ref{1walltqft}) is the $\S^1_L$-dimensional reduction\footnote{$\phi^{(0)}$ is identified with the $\S^1_L$-holonomy of the 3d Chern-Simons gauge field.} of the 3d $U(1)_{N}$ Chern-Simons theory, which has been proposed as a worldvolume theory on $k=1$ walls in $\R^4$ (see \cite{Bashmakov:2018ghn} for recent studies). 
Upon gauging of the $0$-form and $1$-form global symmetries, the theory (\ref{1walltqft}) can be seen to match  the mixed $0$-form center/$1$-form center $\Z_N^{(1),\S^1_L} \Z_N^{(1),\R^2}$ anomaly on its worldvolume, required by anomaly inflow on the $1$-wall between chirally broken vacua.\
 
Deconfinement of quarks on the DW is reflected in the TQFT (\ref{1walltqft}) by the fact that $N$-ality $q$ charges are represented by insertions of $e^{i q \oint a^{(1)}}$ Wilson lines, which obey  perimeter law in the theory  (\ref{1walltqft}). 
Our final comment here is that the TQFT (\ref{1walltqft}) captures (and can be argued to predict) the topological features of deconfinement,\footnote{It would be an interesting exercise to construct the worldvolume TQFTs for the $k>1$ DWs in the semiclassical regime. These TQFTs would have to correctly account for the $ \left( \begin{array}{c}N\cr k\end{array} \right)$ different $k$-walls with fluxes (\ref{fluxes1}, \ref{fluxes2})  and for their arrangements into orbits under the $0$-form center symmetry (\ref{zeroform2}). We shall not attempt this here.}  but does not shed light on the microscopic mechanism, the main focus of our study here.

\section{Numerical studies of BPS and non-BPS Solitons}
\label{dwnumerics}

\subsection{The BPS Equation}
 A domain wall is a co-dimension one static field configuration that interpolates between different vacua. The minimal-energy DW configurations satisfy the BPS equation and are also called BPS solitons.
Let $z$ denotes the one dimensional spatial coordinate. The complex field (\ref{fields1}), which will be also referred as $\xv$ (instead of $\pmb x$), is a vector in the $SU(N)$ Cartan algebra. The components of  $\xv$  will often be labeled as $x_i$, $i=1, \ldots, N-1$. The BPS equation is given by
\begin{eqnarray} \label{eq:BPS}
\dfrac{d\vec{x}}{dz} &=& \dfrac{\alpha}{2} \dfrac{dW^*}{d\vec{x}^*}~, ~ {\rm where}  ~~
\alpha  =  \dfrac{W(\vec{x}(\infty)) - W(\xv(-\infty))}{|W(\xv(\infty)) - W(\xv(-\infty))|}~.  \end{eqnarray}
The boundary conditions on $\xv(\infty)$ and $\xv(-\infty)$ are given by any two distinct vacua (\ref{vacua}), $\xv_k = i {2\pi k \over N} \vec\rho$, $k=0,\ldots N-1$. 

We note that here, and in the rest of the paper, distances are dimensionless and can be thought of as measured in units of the dual photon mass $m$.

\subsection{Solving the BPS Equation Numerically} \label{sec:Method}
\subsubsection{Gauss-Seidel Finite-Difference Method for One Dimension}
We solve the BPS equation numerically by combining the Gauss-Seidel method and the finite-difference method \cite{Numerical_Method}. The finite-difference method converts a boundary-value problem for the Poisson equation to a system of algebraic equations by replacing derivatives with finite differences. The Gauss-Seidel method solves that system of algebraic equations by taking any initial ``guess" of the solution and iteratively ``relaxes" it to the true solution. Because of this, we will also call this method the ``relaxation method."

To convert our problem into a Poisson  equation, we first differentiate the BPS equation to get a second order equation:
\begin{eqnarray}\label{secondorder}
 {d^2 x_i \over d z^2} &=& \frac{1}{4} \sum_{j=1}^{N-1} {\partial W \over \partial x_j}  
 {\partial^2 W^* \over \partial x_j^* \partial x_i^*} \\
&=& \frac{1}{4} \sum_{j=1}^{N-1} \left( \sum_{a=1}^N e^{\ava \cdot \xv} \alpha_a^j \right) \left( \sum_{a=1}^N e^{\ava \cdot \xv^*} \alpha_a^j \alpha_a^i \right)~,  \quad i = 1,\dots,N-1~, \nonumber
\end{eqnarray}
where $\alpha_a^j$ is the $j$-th component of the vector $\ava$.
We discretize this problem by replacing the real  line by finitely many points. We choose a sufficiently large distance to be approximated as ``infinity" and a sufficiently small number, $h$, to be the pixel of the approximation, so that the real line becomes the discretized segment $[z_0, z_f]$: 
 $ L = \{z_0,z_0+h, z_0 +2h, \dots, z_f-h, z_f \}$.

Before describing how to use the method to solve for the vector field $\xv$ on this discrete line, it is simpler to first consider a scalar field. Given any differentiable function $f(z)$, its second derivative can be approximated numerically using the centered-difference approximation:
\begin{equation}
   {d^2 f \over dz^2} (z) = \frac{f(z+h)-2f(z)+f(z-h)}{h^2} + O(h^2)
\end{equation}
Ignoring the remainder $O(h^2)$ by choosing a sufficiently small $h$, we can solve for $f(z)$ in terms of its neighboring values:
\begin{equation} \label{eq:Gass-Seidel}
    f(z) = \frac{f(z+h) + f(z-h) - h^2 f''(z)}{2}
\end{equation}
This formula says that the value of $f$ at the point $z$ is the average of its two neighboring values subtracted by a term proportional to the second derivative at $z$.
In the Gauss-Seidel method, we start with some arbitrary initial field values for each point in the discretized real line $L$, with the exception of the two endpoints, which are fixed to be the desired boundary values. Then for each $z \in L- \{z_0,z_f \}$, we calculate the second derivative of $f$ at $z$, and replace the value of $f(z)$ by the average of its neighboring values minus the second derivative, as in equation (\ref{eq:Gass-Seidel}). We repeat this process until convergence, using the immediately available values during updates (more on this point later). Note that the boundary values $f(z_0)$ and $f(z_f)$ are never changed.

Convergence is measured by how much the field changes from one update to the next. For example, we can define an error function by 
\begin{equation}
    e(f_i) = \frac{1}{n} \sum_{z \in L} |f_{i} (z) - f_{i-1} (z)|
\end{equation}
where $i$ denotes the $i$-th updates and $n$ is the total number of points in the discretization. We can set some small number to be the tolerance for error and only stop the updating process until the error is less than this tolerance ($T$):
\begin{equation}
    e(f_i) < T \quad \text{(stopping condition)}~.
\end{equation}
There are two ways to perform the updating. In the Gauss-Seidel method, as each $f(z)$ is updated,
$$ f(z) \to f_{\text{new}}(z)~, $$
it is immediately used for the computation of its neighbor:
$$ f(z+h) \to \frac{
f_{\text{new}}(z)+ f(z+2h) - h^2 f''(z+h)}{2} $$
In contrast, in the Jacobi method, the new values are retained and are only used for the next iteration:
$$ f(z) \to f_{\text{new}}(z) $$
$$ f(z+h) \to \frac{f(z)+ f(z+2h) - h^2 f''(z+h)}{2} $$
Because the Gauss-Seidel method uses the immediately available values, it is usually considered a better estimate \cite{Numerical_Method}. For our solution to the BPS equation, we found that the Jacobi iteration method converges slowly and we used the Gauss-Seidel method throughout.

To solve for a vector field, we simply apply the above algorithm simultaneously to every component. The only difference is that the second derivative of each component will in general depend on the values of all the other components. Also, the error function gets divided by a further factor of $N-1$, which is the number of fields, so that it is comparable across cases with different $N$.

\subsubsection{Verifying the BPS Nature of Solutions} \label{sec:verification}
The Gauss-Seidel finite-difference method gives us the solution to the second order equation. We still need to check that it solves the first order BPS equation. We accomplish this by comparing the numerical derivative and theoretical derivative of our solution field.

Let $\xv$ be a solution to the second order equation (\ref{secondorder}). We can numerically compute its first derivative with respect to $z$ simply by
\begin{equation} \label{eq:numerical_first_derivative}
 {d\xv \over d z} (z)\vert_{\rm 2nd\;  order} = \frac{\xv(z+h) - \xv(z-h)}{2h}
\end{equation}
If this is also a solution to the BPS equation, then it must also satisfy \begin{equation}
     {d\vec{x} \over dz}\vert_{\rm ``theor."} = \frac{\alpha}{2} \frac{dW^*}{d\vec{x}^*} = \frac{\alpha}{2} \sum_{a=1}^N e^{\ava \cdot \xv^*} \ava \label{eq:theoretical_first_derivative}
\end{equation}
We call the above equation the ``theoretical" first derivative, where we use the solution of the 2nd order equation to compute the r.h.s.

We substitute the $\xv$ we got by solving the 2nd order equation (\ref{secondorder}) into equation (\ref{eq:numerical_first_derivative}) and (\ref{eq:theoretical_first_derivative}) and plot them together. If we have a minimal energy solution (BPS solution), then the two plots should coincide. If we have a non-BPS solution that only solves the second order equation but not the first order equation, then the two graphs will look different.

As a further check, we also compute the energy of a solution. The dimensionless energy per unit length  of any static field configuration  (BPS or not)  interpolating between two vacua is given by \cite{Hori:2000ck}
\begin{equation} \label{eq:Energy}
    E = \int_{-\infty}^{\infty} dz \left\{ \left|  {d \xv \over dz} \right|^2 + \frac{1}{4} \left|  {d W \over d \xv} \right|^2 \right\}~.
\end{equation}
On the other hand, a BPS soliton, by definition, has energy
\begin{equation} \label{eq:BPS Energy}
    E = \left| W(\xv(\infty)) - W(\xv(-\infty)) \right|~.
\end{equation}
So given a solution to the second order equation, we substitute it into equation (\ref{eq:Energy}) and (\ref{eq:BPS Energy}) and see if they are the same. We perform the numerical integration in equation (\ref{eq:Energy}) using the trapezoid rule.

\subsubsection{Non-Uniqueness of Solutions} \label{sec:non-unqiueness}

In general, there is no uniqueness theorem for boundary value problems. For the BPS equation, one kind of non-uniqueness is the location of the center of the solution. This can be understood as a result of the property that the BPS equation is invariant under the transformation $z \to z - z_0$, where $z_0$ is any constant, leading to a freedom in choosing  $z_0$.
Generally, choosing different initial values in the relaxation process will result in solutions with different $z_0$. For easy visualization and consistency, we force our solutions to be centered at the origin by choosing the following set of initial values for the imaginary component of $\xv(z)$:
\begin{equation}
    \sigmav(z) = \begin{cases}
    \sigmav(-\infty), \quad z< 0 \\
    \sigmav(\infty), \quad z \geq 0
    \end{cases}
\end{equation}
This guarantees that the initial values already have an artificial ``kink" at the origin.\footnote{As for the real component, for all vacua (\ref{vacua}) $\phiv(-\infty) = \phiv(\infty)=0$. We can either start with setting $\phi$ equal to zero on the boundaries and some nonzero value in between, or have $\phi$ be initially zero throughout. }
 If a kink solution exists, the Gauss-Siedel iteration will smooth this out into a correct solution. Of course, we make sure that choosing any other initial values for $\sigmav(z)$ will still yield the same result, albeit centered at a different location.
We believe this translational freedom is the only source of non-uniqueness of the BPS solutions.

\subsubsection{An Example of the Method}
We show an example here that demonstrates how this method works. Here, we solve for the $k=1$ wall in $SU(3)$,   with boundary conditions $\xv(-\infty) = 0$ and $\xv(\infty) = i \frac{2 \pi}{3} \rhov$. Our range is from $z_0 = -20$ to $z_f = 20$, with $200$ points in between. This gives a lattice spacing of $h = 0.2$.\footnote{In units of the dual photon mass $m$, the length of our interval is $40/m$, the pixel size is $0.2/m$.  We note that the width of the DW (see Fig.~\ref{fig:SU(3) k=0 to k=1}), is determined by the inverse of the  lowest of the dual photon masses, $\sim m$ for $SU(3)$.}
The initial values  are set up such that the final result will have the kink centered at the origin, as described in section \ref{sec:non-unqiueness}.

%\begin{figure}[ht]
   % \centering
    %\includegraphics[scale=0.6]{"SU(3) Initial Values Example".png}
    %\caption{An example for initial values taken for $SU(3)$, $\xv(-\infty) = 0$, $\xv(\infty) = i\frac{2\pi}{N} \rhov$. Note that the real component has much smaller initial values.}
 %   \label{fig:Initial Values}
%\end{figure}

The relaxation result, with tolerance set to $T = 10^{-9}$, is shown in figure \ref{fig:SU(3) k=0 to k=1}. The top two pictures are the real and imaginary component of the solution, respectively. The bottom two pictures are their corresponding derivative with respect to $z$. Note that both of the numerical and theoretical derivatives were plotted, with the theoretical curve being plotted in dotted line. However, the dotted lines are not visible in the pictures, because the theoretical and numerical derivative completely overlap. This shows that the solution we found also solves the first order equation and is a BPS soliton.

\begin{figure}[ht]
    \centering
    \includegraphics[scale=0.6]{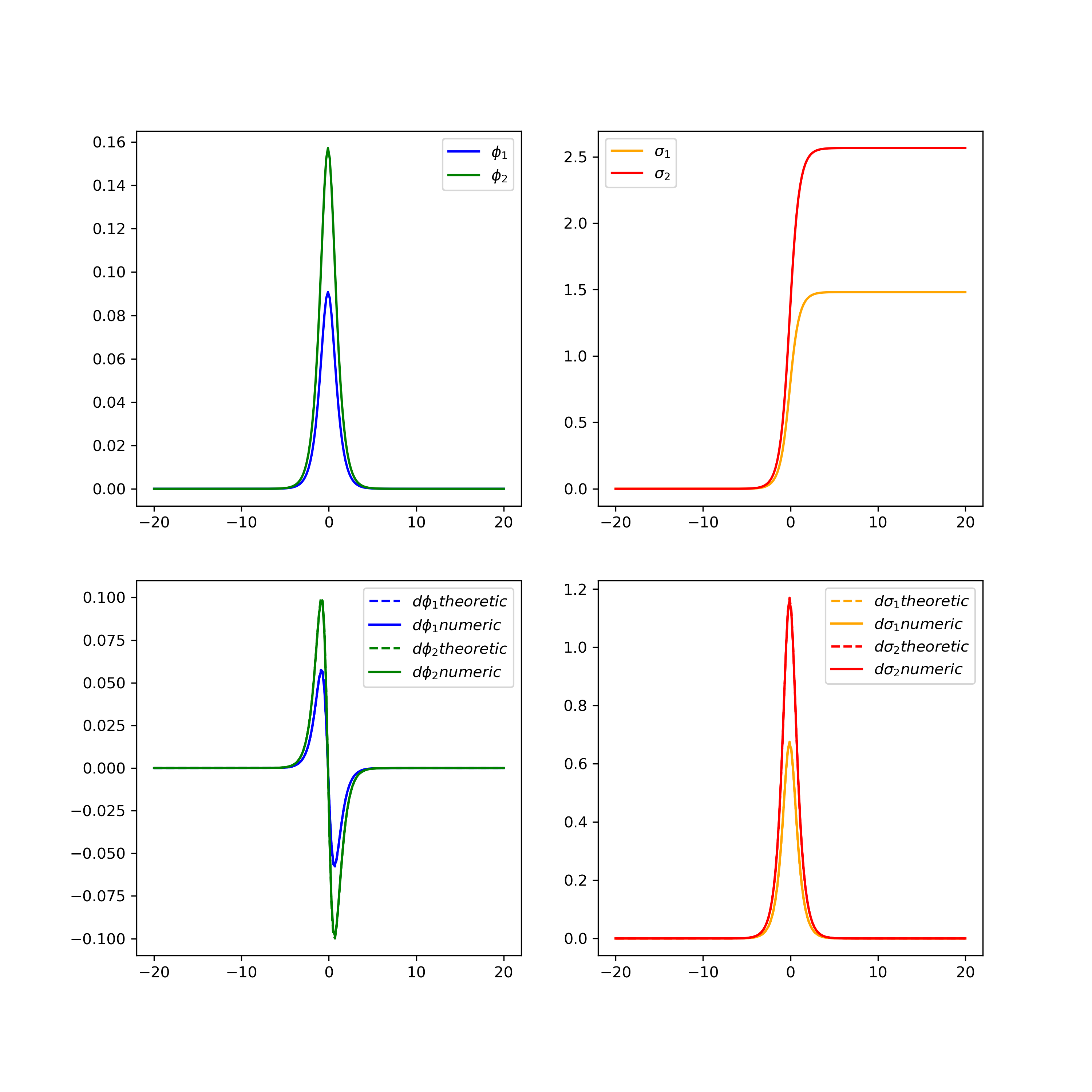}
    \caption{BPS solutions for $SU(3)$, $\xv(-\infty) = 0$, $\xv(\infty) = i\frac{2\pi}{N} \rhov$. The grid goes from $-20$ to $20$. The pixel is $h=0.2$. The tolerance is $T = 10^{-9}$.}
    \label{fig:SU(3) k=0 to k=1}
\end{figure}

In contrast, an example of a boundary condition that admits no BPS solution is $\xv(-\infty) = 0$, $\xv(\infty) = i 2 \frac{2\pi}{N} \rhov$. This non-BPS solution is shown in figure \ref{fig:SU(3) k=0 to k=2}. We know this is a non-BPS solution since the dotted-lines show up in the plot, meaning that the numerical and ``theoretical" first derivatives (\ref{eq:numerical_first_derivative}) and (\ref{eq:theoretical_first_derivative}) do not match. 
\begin{figure}[ht]
    \centering
    \includegraphics[scale=0.6]{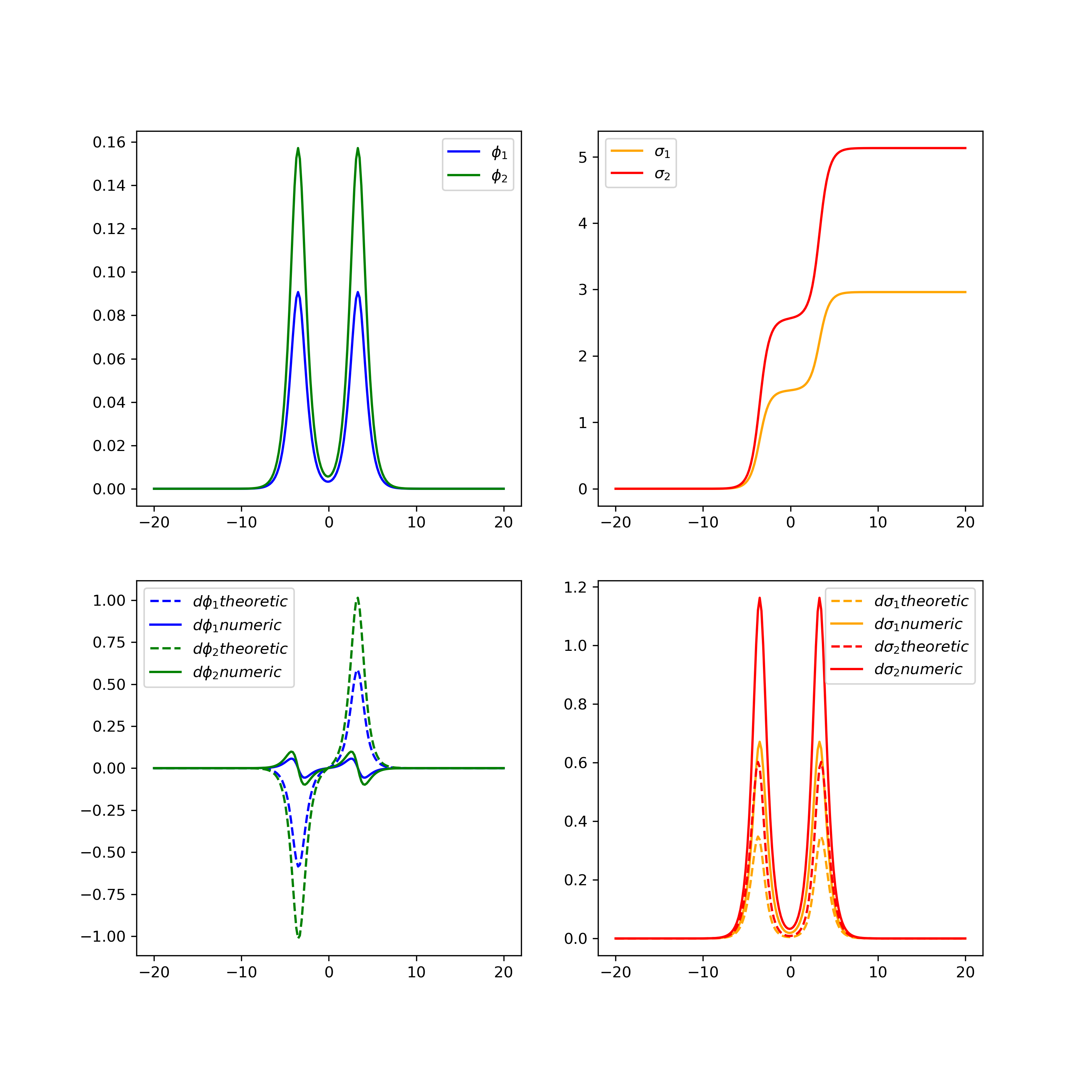}
    \caption{Non-BPS solutions for $SU(3)$,  $\xv(-\infty) = 0$, $\xv(\infty) = i 2 \frac{2\pi}{N} \rhov$. The grid goes from $-20$ to $20$. The pixel is $h=0.2$. The tolerance is $T = 10^{-6}$.}
    \label{fig:SU(3) k=0 to k=2}
\end{figure}

\subsection{Notable Results}
The method we described in Section \ref{sec:Method} can solve any BPS equation with any given boundary conditions, as well as distinguish the boundary conditions  that do not admit BPS solutions. We now describe some of the results we find worthy of noting, found using this method.

\subsubsection{The Formula for Boundary Conditions that Admit BPS Solutions}
The key result that we found is a formula for the boundary conditions that admit BPS solutions. We found strong numerical evidence for this formula and hypothesize that this is true in general.

Let's call any BPS soliton that connects two vacua $k$ units apart in the $\sigma$-space a $k$-wall. It is easy to show that any $k$-wall in $SU(N)$ is equivalent, up to a sign and the addition of a constant vector, to one that starts from a vacuum of the unit cell that is equivalent to the origin. Among these, we found that the boundary conditions that admit BPS solutions are given by
\begin{align} \label{eq:N choose k}
\begin{cases}
    i 2\pi \wv_{i_1} + \dots + i 2 \pi \wv_{i_{k-1}} &\to i k \frac{2\pi}{N} \rhov \\
    i 2\pi \wv_{i_i} + \dots + i 2 \pi \wv_{i_{k}} &\to i k \frac{2\pi}{N} \rhov
\end{cases}
\end{align}
where $\wv_j$ are the fundamental weights and the sets $i_j$ in each case are different numbers taken from $1, \dots, N-1$.
Since there are only $N-1$ different fundamental weights, equation (\ref{eq:N choose k}) implies that the total number of $k$-walls in $SU(N)$ is
\begin{equation}
    {N-1 \choose k-1} + {N-1 \choose k} = {N \choose k}
\end{equation}
This counting of BPS solutions agrees with the previously known result \cite{Hori:2000ck}.

Although this pattern was originally noticed using the method detailed in Section \ref{sec:Method}, where we manually check whether the theoretical and numerical first derivatives match for each case, in order to determine the BPS-ness of a solution, we later develop a more efficient method to test this formula for a large number of cases. For each $N$ and each $k$, we first compute all the solutions to the second order equations with boundary conditions $k$ units apart. Then we compute their numerical energy, using the method described in section \ref{sec:verification}. We plot all of their energy together, along with the theoretical minimal energy of a BPS solutions, given by equation (\ref{eq:BPS Energy}). Since all non-BPS solutions have energy that are larger than the BPS energy, it is easy to tell from such plot which boundaries give rise to BPS solutions.
We check this for all $N$ and $k$ up to $N=8$. So far, all results conform to our hypothesis. We show an example for $SU(6)$ and $k=3$ in figure \ref{fig:SU(6) k=3 Walls}.

\begin{figure}[ht]
    \centering
    \includegraphics[scale=0.2]{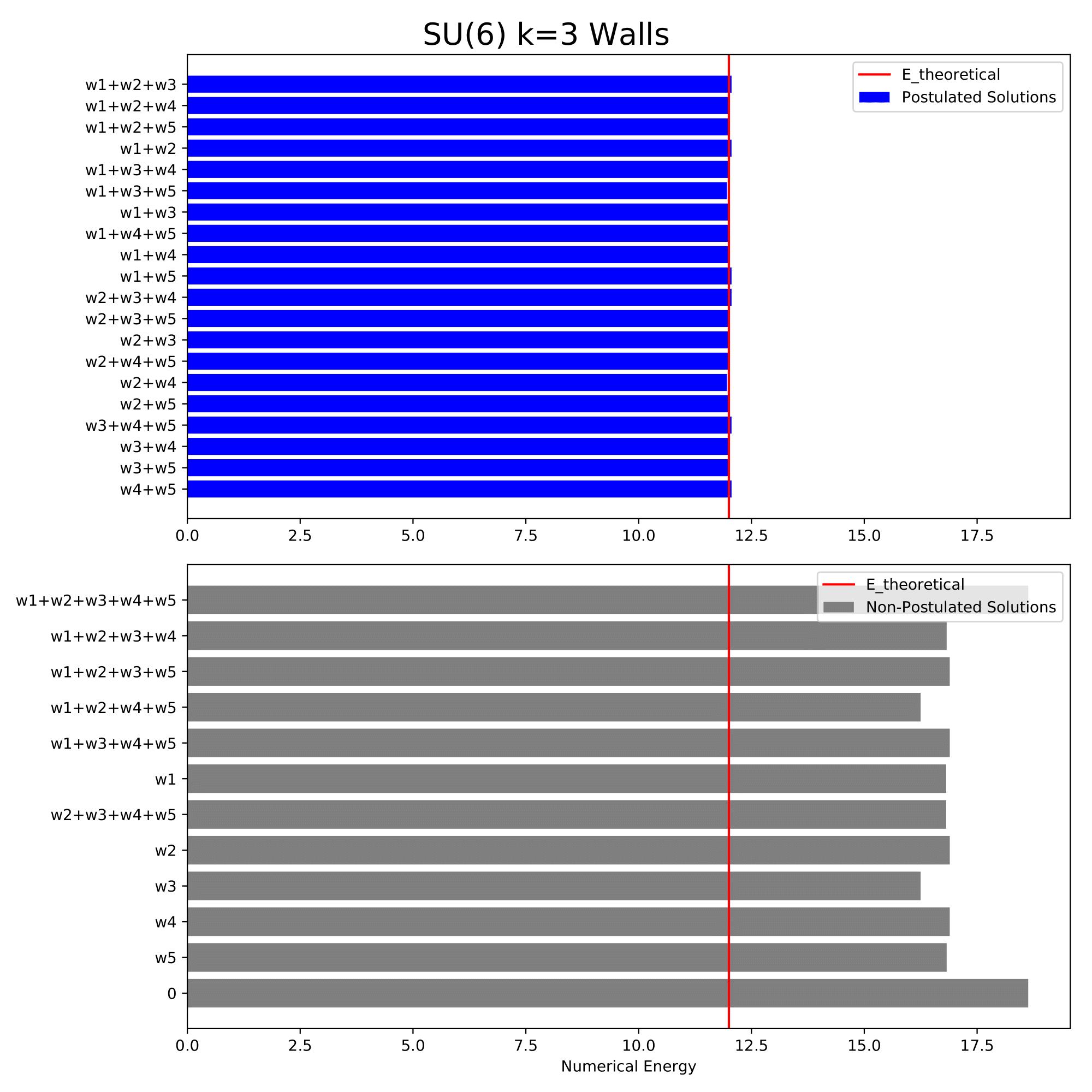}
    \caption{Energy for $SU(6)$ $k=3$ walls. The blue bars are the energies of solutions that Eqn.~(\ref{eq:N choose k}) claims to be the BPS solutions. The grey bars are the energies of all other solutions. The red vertical line is the  energy of BPS solutions (\ref{eq:BPS Energy}). This graph, as well as analogous results for other $N$ and $k$, strongly supports Eqn.~(\ref{eq:N choose k}).}
    \label{fig:SU(6) k=3 Walls}
\end{figure}

Furthermore, as described in Section~\ref{dwdeconfinement}, we found that the group of the sum of $k$ or $k-1$ number of fundamental weights has a special property with respect to the zero-form $\mathbb{Z}_N$ center symmetry, lending further  support to the argument that this equation gives all possible boundaries for BPS solitons.

\subsubsection{Reversed Direction of Electric Fields}
For $N < 5$, the qualitative behaviour of the BPS dual photons in the DWs are the same: every dual photon starts at a vacuum at negative infinity, increases abruptly at a kink, reaches an inflection point and quickly starts decreasing, and then plateaus to a second vacuum (or the other way around: starts from a higher vacuum and decreases through a kink to a lower vacuum). A good example of this is the $SU(3)$ case shown in Figure~\ref{fig:SU(3) k=0 to k=1}. In particular, every dual photon has one inflection point, for $N<5$. 

However, there appears to be a change starting at $N=5$. Here, with seemingly no discernible patterns on the associated boundary conditions, some of the BPS solutions have some components of its dual photon possessing three inflection points. This means its behaviour is qualitatively different: a dual photon of this kind starts at a vacuum at negative infinity, decreases slightly before hitting the kink, then increases until it hits a second inflection point, slows down to a new plateau, and then dips a little bit, before finally settling down to a new vacuum. In other words, these dual photons ``dip" a little right before and after the kink.

One example of this is a $k=1$ wall in $SU(6)$, with boundary conditions $\xv(-\infty) = 0$, $\xv(\infty) = i \frac{2 \pi}{N} \rhov$,  shown in Figure~\ref{fig:SU(6) k=0 to k=1}. In this case, only $\sigma_1$ exhibits this ``double-dipping" phenomenon, while all other components of $\sigma$ only have one inflection point. Interestingly, the corresponding $\phi_1$ also stands out from the crowd. All the other $\phi$ have the same peak shape, while $\phi_1$ behaves differently.

\begin{figure}[ht]
    \includegraphics[scale=0.5]{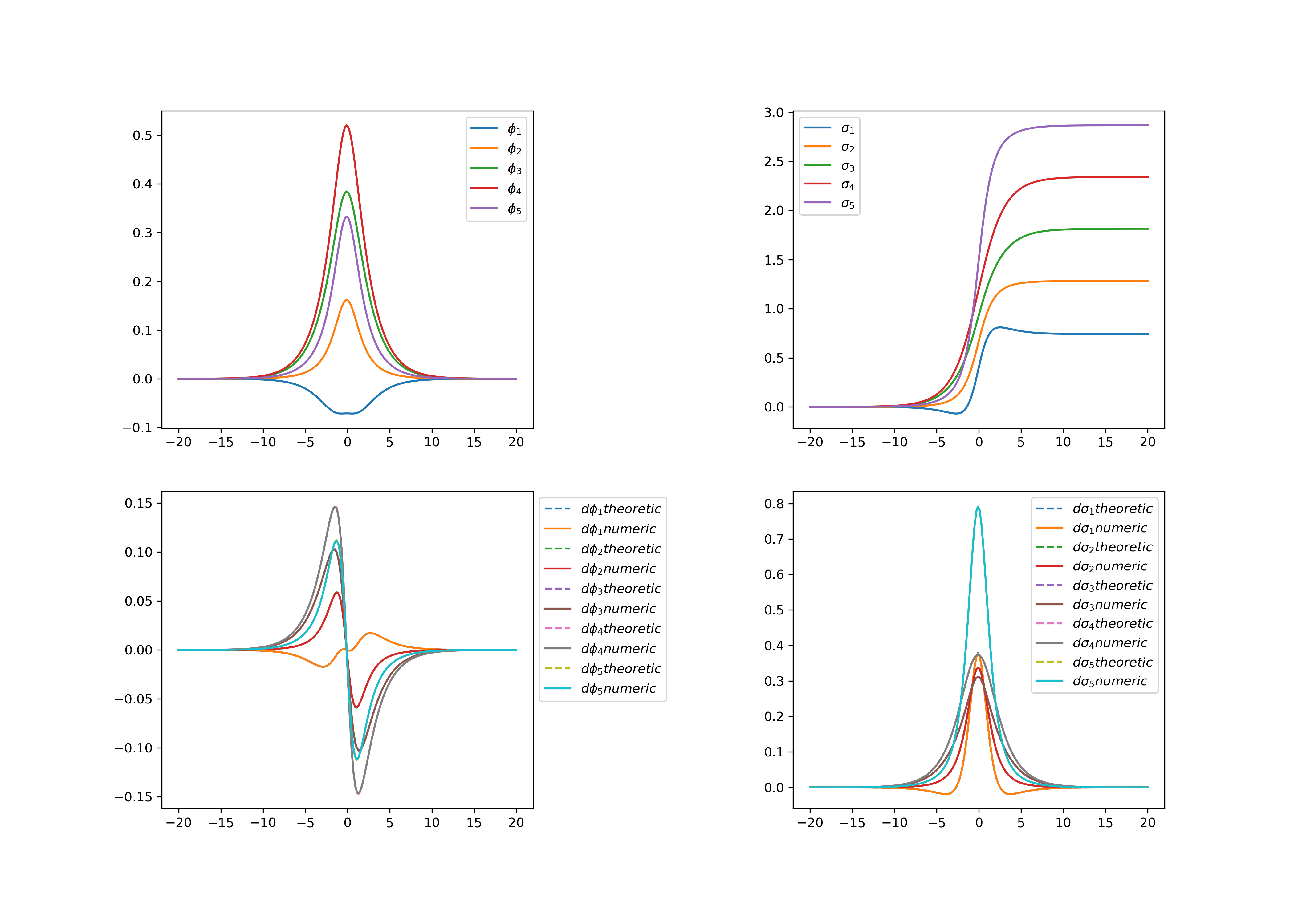}
    \caption{``Double dipping" in $\sigma_1$ and an unusual shape in $\phi_1$.  BPS solutions for $SU(6)$, $\xv(-\infty) = 0$, $\xv(\infty) = i\frac{2\pi}{N} \rhov$. The grid goes from $-20$ to $20$. The pixel is $h=0.2$. The tolerance is $T = 10^{-6}$.}
    \label{fig:SU(6) k=0 to k=1}
\end{figure}

Another example is a $k=2$ wall in $SU(5)$, with boundary conditions $\xv(-\infty) = i 2 \pi \wv_3 $, $\xv(\infty) =i 2 \frac{2\pi}{N} \rhov$, as shown in Figure~\ref{fig:SU(5) k=w3 to k=2}. Here, there are two components that exhibit this behaviour and their corresponding magnetic fields are equally exotic.

\begin{figure}[ht]
    \includegraphics[scale=0.6,width=1.1\textwidth]{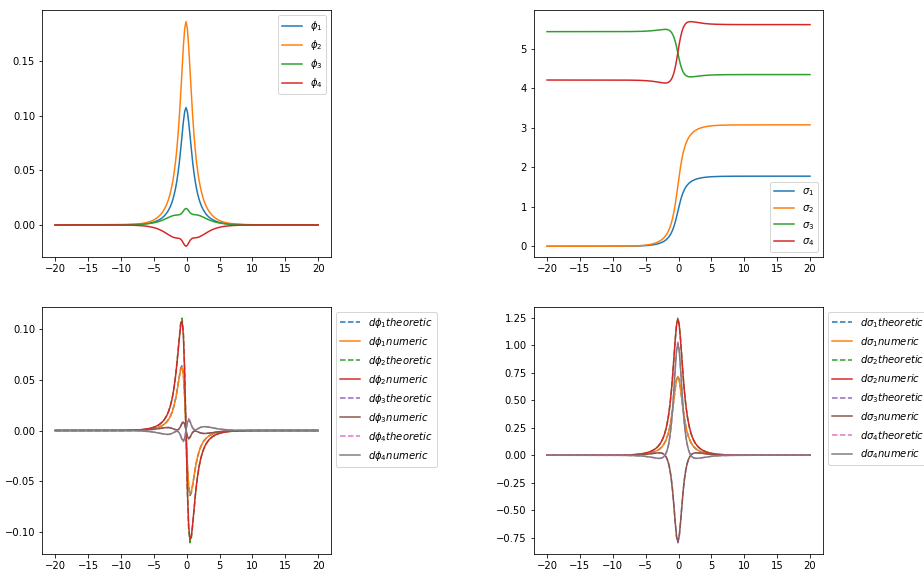}
    \caption{``Double dipping" in $\sigma_3$ and $\sigma_4$ and unusual shapes in $\phi_3$ and $\phi_4$. BPS solutions for $SU(5)$,  $\xv(-\infty) = i 2 \pi \wv_3$, $\xv(\infty) =i 2 \frac{2\pi}{N} \rhov$. The grid goes from $-20$ to $20$. The pixel is $h=0.2$. The tolerance is $T = 10^{-6}$.}
    \label{fig:SU(5) k=w3 to k=2}
\end{figure}

In all of these cases, notice that a ``double-dipping" $\sigma$ field has a derivative that changes sign at two different points. The physical implication is that the corresponding electric field reverses its direction at two points. For example, in $ SU(5) $, consider a quark and an anti-quark with weights $ \pm 2 \pi \wv_3 $ that are deconfined between $k=2$ BPS-walls, as shown in Figure~\ref{fig:Deconfinement Reverse Electric}. The dual photons interpolate the vacua in the $z$-direction, and  sufficiently far  from the quarks have no dependence on the $x$-directions. By definition, the derivative of a dual photon in the $z$-direction is equal to the electric field in the $x$-direction (up to a constant):
\begin{equation}
  { d \sigma_i  \over d z} \propto E^i_x
\end{equation}
This effectively means we can rotate the dual photon field by 90 degrees to get the picture of the electric field. 

In between the quarks, the dual photon is a BPS solution that corresponds to Fig.~\ref{fig:SU(5) k=w3 to k=2}. The fact that $ {d \sigma_4 \over d z}$ changes sign at two points means that the electric field lines point toward the positive quark at two vertical positions, even though the majority of the field lines point away from it. This is the case for $E^4_x$, since $\sigma_4$ has the ``double dipping" shape.
This is unusual because normally, $ {d \sigma_i \over d z}$ does not change sign and all of the electric field lines point away from the positive quark and toward the negative quark.

\begin{figure}[ht]
    \includegraphics[scale=1.5]{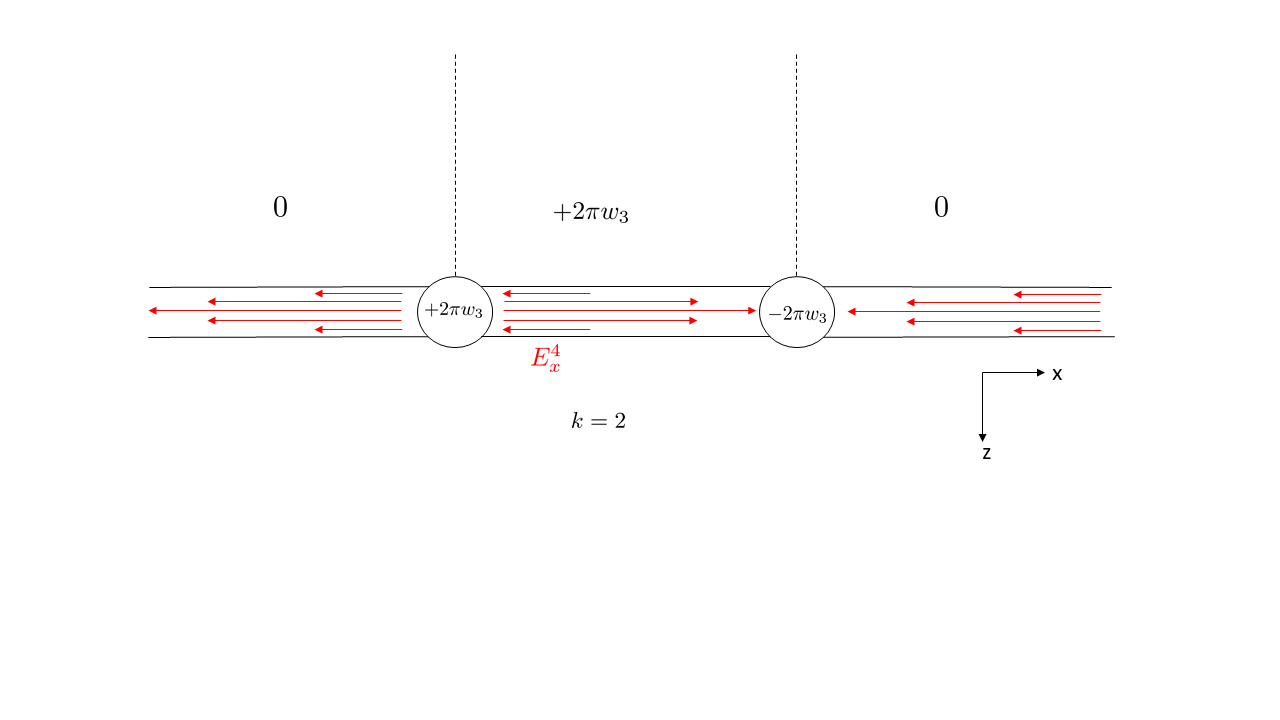}
    \caption{A quark-anti-quark pair with weights $\pm 2 \pi \wv_3$ deconfined on a $k=2$ BPS-wall in $SU(5)$. Only the fourth component of the electric field is plotted (in red). Note that in between the quarks, some of the field lines point away from the negative quark and toward the positive quark. This is a result of the ``double dipping" shape of the $\sigma_4$ solution. The field lines on the two sides do not have this property because their $\sigma_4$ solutions do not have this shape (not shown).}
    \label{fig:Deconfinement Reverse Electric}
\end{figure}

\subsubsection{``Magnetless" Solitons}
It can be shown \cite{Anber:2015kea} that in $SU(2)$, the BPS soliton has no magnetic field along its worldvolume, i.e. $\phiv = 0$. This is only a special case. For higher $N$, $\phiv$ is non-trivial in general. We shall call solutions where $\vec\phi \equiv 0$ along the DW profile ``magnetless."

However, numerical results show that for $N=4$, all of the $k=2$ walls have $\phiv = 0$. An example is shown in Figure \ref{fig:SU(4) k=w2 to k=1}. (Note that the dotted lines for $ {d \phiv \over d z}$ do not get covered up here; this is not a problem since the magnitude is of order $10^{-7}$. The error is just numerical error.) We still know that this is a BPS solution since the numerical energy matches the theoretical energy.
\begin{figure}[ht]
    \centering
    \includegraphics[scale=0.6]{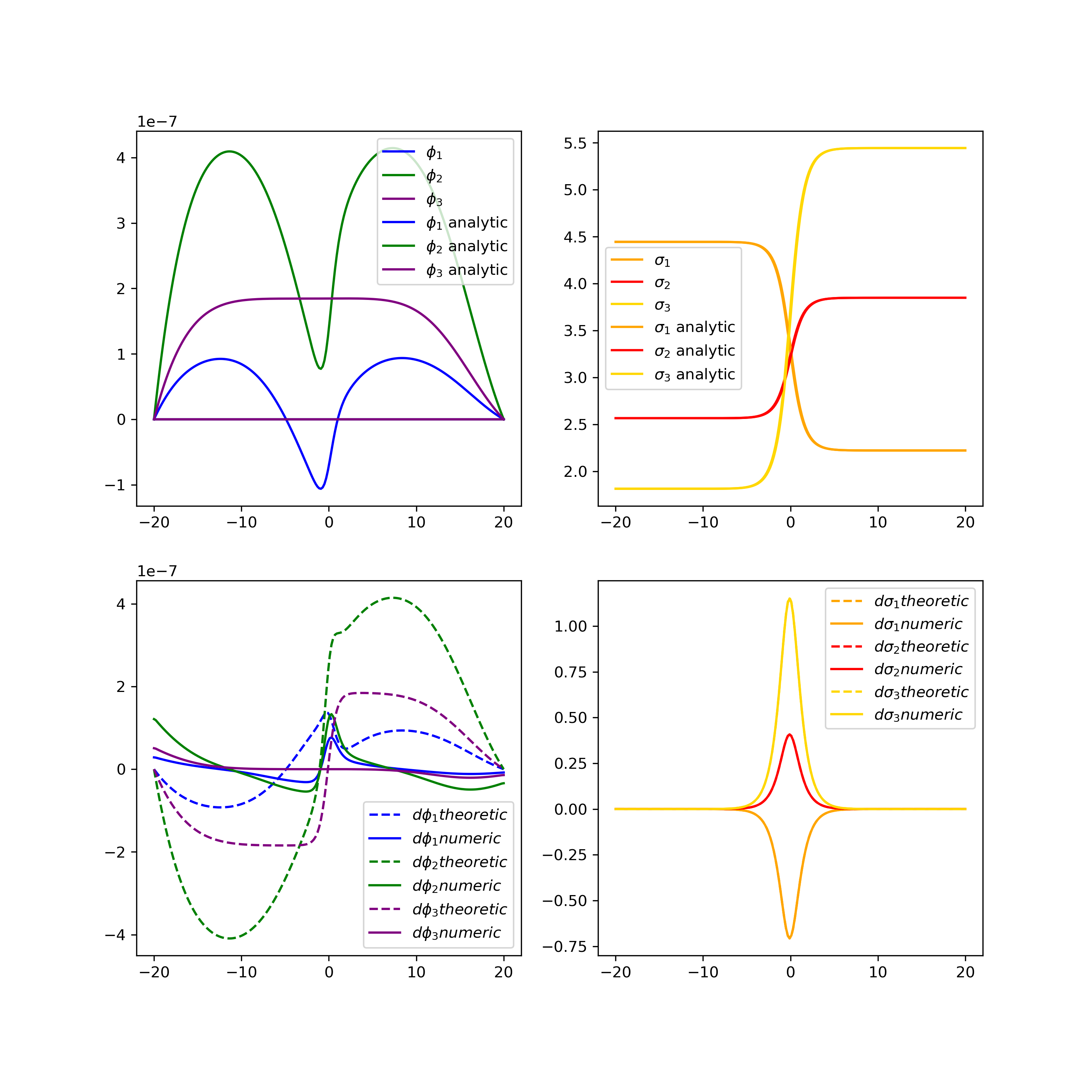}
    \caption{BPS solutions for $SU(4)$, $\xv(-\infty) = i2\pi \wv_1$, $\xv(\infty) = i2\frac{2\pi}{N} \rhov$. The grid goes from $-20$ to $20$. The pixel is $h=0.2$. The tolerance is $T = 10^{-8}$.}
    \label{fig:SU(4) k=w2 to k=1}
\end{figure}
The fact that the numerical solution for $\phi$ is so small signals that the true solution may be identically zero. Indeed, taking the hint from the numerical result, we find the analytic solutions for all $SU(4)$ $k=2$ solitons and explicitly show that all of their $\phiv$ are identically zero. The derivation is shown in Appendix \ref{magnetless}.

Furthermore, we found that the fact that both the $k=1$ walls for $SU(2)$ and $k=2$ walls for $SU(4)$ are magnetless is  no coincidence.  We found an argument that shows that only $SU(N)$ theories with even $N$ can admit magnetless solutions, which occur only for  $k=\frac{N}{2}$-walls. Furthermore, we find an explicit formula for all boundary conditions that give rise to magnetless solitons, as well as the their analytic solutions. The detail derivation is given in Appendix \ref{magnetless}. In short, if $N$ is divisible by $4$, then there are $6$ magnetless solitons. If $N$ is not divisible by $4$ (but is even), there are $2$ magnetless solitons.

\section{Domain Walls  with Quarks}
\label{dwbaryons}
Suppose we have a configuration with $n$ quarks, each with charges $\vec{c}_i$  (elements of the weight lattice of $SU(N)$) at positions $\vec{r}_i = (y_i, z_i)$, $i = 1, \ldots, n$. Then the source part of the action in the electric description of the theory is given by
\begin{equation}
    S_{{\rm source}}=- \int dt \sum_{i}\left(\vec{c}_i\right)_a A_{0}^{a}(\vec{r}_i).
\end{equation}
Notice that we can write $A_0^a(\vec{r}_i)$ as an integral from $y=y_i$ to $\pm\infty$, evaluated at $z=z_i$,
\begin{equation}
    A_0^a(\vec{r}_i)=-\int_{y_i}^{\pm\infty} {d y}\partial_y A_0^a(z_i,y)=\int_{y_i}^{\pm\infty} {d y}F_{0y}^a(z_i,y),
\end{equation}
where we have assumed that $A_0^a$ vanishes at infinity. We can then use the duality map (\ref{fieldrelation}) to the magnetic description to get
\begin{equation}
    A_0^a(\vec{r}_i)=-\frac{g^2}{4\pi L}\int_{y_i}^{\pm\infty} {d y}\partial_z\sigma^a(z,y)\mid_{z=z_i}.
\end{equation}
Finally, the source term in the action is
\begin{equation}
    S_{{\rm source}}= \frac{g^2}{4\pi L}\int {d t}\sum_i\left(\vec{c}_i\right)_a\int_{y_i}^{\pm\infty} {d y}\partial_z\sigma^a(z,y)\mid_{z=z_i}.
\end{equation}
Writing this as an integral over all space, we get
\begin{eqnarray}\label{source1}
    S_{\text{source}}&=& \frac{g^2}{4\pi L} \int dt \int dz \int dy \sum_i\left(\vec{c}_i\right)_a\delta(z-z_i)\int_{y_i}^{\pm\infty} {d y'}\delta(y-y')\partial_z\sigma^a(z,y)\\
    &=&- \frac{g^2}{4\pi L}\int {d t}\int {d z}\int {d y}\sum_i\left(\vec{c}_i\right)_a\left[\partial_z\delta(z-z_i)\right]\int_{y_i}^{\pm\infty} {d y'}\delta(y-y')\sigma^a(z,y), \nonumber
\end{eqnarray}
where going from the first line to the second we used integration by parts. 

Then, the full action, recalling that the full field we are describing is a complex field with components $x^a=\phi^a+i\sigma^a$, is given by\footnote{Note that upon going from (\ref{source1}) to (\ref{abc}), we made the  spacetime coordinates $(t,y,z)$ in (\ref{abc})   dimensionless, rescaling them by the dual photon mass $m$ (compared to the coordinates in (\ref{source1}) and earlier in this Section; for brevity, they are still denoted by the same letters, hoping to not  cause undue confusion). Further, we restored the parameters $M$ and $m$ to conform with the SYM Lagrangian as given earlier  (\ref{lagrangian1}).}
\begin{eqnarray}
    S&=& {M \over m}\int dt \int d z \int {d y}\left(\left\lvert\partial_0 x^a\right\rvert^2-\left\lvert\partial_i x^a\right\rvert^2-\frac{1}{4}\left\lvert {d W \over d {x}^a}\right\rvert^{2}-\right. \label{abc} \\
    &&\qquad \qquad \qquad \qquad \left.4\pi\sum_i\left(\vec{c_i}\right)_a\partial_z\delta(z-z_i)\int_{y_i}^{\pm\infty} {d y'}\delta(y-y')\operatorname{Im}\left[x^a(z,y)\right]\right)\nonumber, 
\end{eqnarray} where a sum over $a$ is implicit.
If we look at the static case, the equations of motion are then
\begin{equation}
    \nabla^2 {x^a}=\frac{1}{4} {\partial \over \partial \left(x^*\right)^a}\left\lvert {d W \over d\vec{x}}\right\rvert^2+2\pi i\sum_j\left(\vec{c}_j\right)_a\partial_z\delta(z-z_i)\int_{y_i}^{\pm\infty} {d y'}\delta(y-y').\label{eqn:eom_2d}
\end{equation}
\subsubsection{Numerical Implementation of Dirac Delta Derivative}
In order to implement equation \eqref{eqn:eom_2d} numerically, we must figure out how to implement the derivative of the Dirac delta function numerically. We use the sifting property of the Dirac delta and integration by parts to get the following general property,
\begin{equation}
    \int dx f(x)\partial_x\delta(x-x_i)=-\partial_x f(x)\mid_{x=x_i}.
\end{equation}
Now we want to discretize this. Let $g(x)=\partial_x\delta(x-x_i)$. Suppose we have discretized our space into a grid with spacing $h$, and the points are labelled by $x_k=x_0+kh$. Then, the above relation becomes
\begin{equation}
    \sum_k hf(x_k)g(x_k)=-\frac{1}{h}\left[f(x_i)-f(x_{i-1})\right].
\end{equation}
Since this must hold for all functions $f$, we find that on the left hand side, only terms with $k=i$ and $k=i-1$ contribute, thus we have
\begin{equation}
    h\left[f(x_i)g(x_i)+f(x_{i-1})g(x_{i-1})\right]=-\frac{1}{h}\left[f(x_i)-f(x_{i-1})\right]
\end{equation}
from which we get the following
\begin{equation}
    g(x_k)=\begin{cases}-\frac{1}{h^2} & k=i\\\frac{1}{h^2} & k=i-1\end{cases}=\frac{1}{h^2}\left(\delta_{k,i-1}-\delta_{k,i}\right).
\end{equation}
Notice that we used the one sided derivative approaching from the left, we could have used the one sided derivative approaching from the right, and we would obtain a similar result, it is just a matter of convention. Using the two sided derivative would result in having the ``jump" spread across three points instead of two, and thus would not be ideal, especially with a coarse grid.

\subsubsection{Relaxation Method}
In two dimensions, one can show that using finite differences the Laplacian of a function $f$ at point $(z_i,y_i)$ is given by
\begin{equation}
    \nabla^2 {f(z_i,y_i)}=\frac{f(z_{i+1},y_i)+f(z_i,y_{i+1})-4f(z_i,y_i)+f(z_{i-1},y_i)+f(z_i,y_{i-1})}{h^2}+\mathcal{O}(h^2).
\end{equation}
We can then rearrange this to find an expression for $f(z_i,y_i)$,
\begin{equation}
    f(z_i,y_i)=\frac{1}{4}\left[f(z_{i+1},y_i)+f(z_i,y_{i+1})+f(z_{i-1},y_i)+f(z_i,y_{i-1})-h^2\nabla^2 {f(z_i,y_i)}\right]+\mathcal{O}(h^4).
\end{equation}
Similar to the 1D case, the relaxation method in two dimensions starts with a grid of points with some boundary conditions, and updates each point iteratively until the solution converges. If $f^k(z_i,y_i)$ is the value of the solution at the $k^{th}$ iteration, then using Gauss-Siedel relaxation it will be updated to be
\begin{equation}
    f^{k+1}(z_i,y_i)=\frac{1}{4}\left[f^k(z_{i+1},y_i)+f^k(z_i,y_{i+1})+f^{k+1}(z_{i-1},y_i)+f^{k+1}(z_i,y_{i-1})-h^2\nabla^2{f^k(z_i,y_i)}\right],\label{eqn:relax_2d}
\end{equation}
where we can see the points that come before $(z_i,y_i)$ have already been updated and are being used. The Laplacian is not the one calculated from finite differences, but some computable function, which for us is given by the discretized version of equation \eqref{eqn:eom_2d}.

\subsection{Solving for Quark Deconfinement on Domain Walls in $\mathbf{SU(2)}$}
\paragraph{}
Here we focus on $SU(2)$ with two opposite charges, $c_1=1/\sqrt{2}$ and $c_2=-1/\sqrt{2}$ (the weights of the fundamental representation in our normalization), located at $z=\pm r$ and $y=0$. The superpotential is given by $W=e^{\sqrt{2}x}+e^{-\sqrt{2}x}$. Thus the equations of motion are
\begin{equation}
    \nabla^2 {x}=\sqrt{2}\left(\sinh\left(2\sqrt{2}\phi\right)+i\sin\left(2\sqrt{2}\sigma\right)\right)+\sqrt{2}\pi i\int_{0}^{\pm\infty} {d y'}\delta(y-y')\left[\partial_z\delta(z+r)-\partial_z\delta(z-r)\right].
\end{equation}
We notice that we can split this into two equations, one for $\phi$ and one for $\sigma$. Doing so, we see that the equation for $\phi$, $\nabla^2{\phi}=\sqrt{2}\sinh\left(2\sqrt{2}\phi\right)$, is consistent if we set $\phi=0$, so we do this and focus only on $\sigma$. Thus, the equation of interest is the following
\begin{equation}
    \nabla^2{\sigma}=\sqrt{2}\sin\left(2\sqrt{2}\sigma\right)+\sqrt{2}\pi \int_{0}^{\pm\infty} {d y'}\delta(y-y')\left[\partial_z\delta(z+r)-\partial_z\delta(z-r)\right].
\end{equation}
Notice that we left the upper limit in the integral over $y$ as $\pm\infty$, we did this to have a more general result, but now we will choose the positive sign. The sign we choose simply informs us on where there will be a jump in $\sigma$, and should align with our boundary conditions, discussed next.

With a box of size $L\times H$ (i.e., with $z$ ranging from $-L/2$ to $L/2$, $y$ from $-H/2$ to $H/2$), we set our boundary conditions to be
\begin{align}
    \sigma(z,-H/2)&=\pi\rho=\frac{\pi}{\sqrt{2}}
    \end{align}
    \begin{align}
    \sigma(z,H/2)&=\begin{cases}
    0 & -L/2\leq z< -r\\
    2\pi w_1=\sqrt{2}\pi & -r\leq z<r\\
    0 & r\leq z< L/2\end{cases}\\
    \sigma(-L/2,y)&=\sigma(L/2,y).
\end{align}
Notice that the last boundary condition simply specifies periodic boundary conditions in $z$. Thus, after setting the first two boundary conditions, we update starting from $(-L/2,-H/2+h)$ and ending on $(L/2-h,H/2-h)$. The points along $z=L/2$ are updated similarly with equation \eqref{eqn:relax_2d}, where we have identified $\sigma(L/2+h,y)$ with $\sigma(-L/2,y)$.
 
We expect the solution to look something like the 1D BPS soliton for $SU(2)$ along the $y$ direction, interpolating between the boundaries at $y=\pm H/2$. Thus, we choose our initial conditions to be just that. Since we know the exact solution for the 1D equations for $SU(2)$, $\sigma(y)=\pm\sqrt{2}\arctan\left(e^{-2y}\right)+\sigma(\infty)$, where the sign of the first term is determined by whether $\sigma(-\infty)$ is greater than $\sigma(\infty)$ and we have a minus sign in the argument of the exponential because we have chosen to put the $k=1$ vacuum at $-\infty$ and the $k=0$ vacuum at $+\infty$. Further, at $z=\pm r$ we expect a more sudden jump in $\sigma$, so we set $\sigma$ to initially be a step function at those points. Thus the initial grid is specified by
\begin{align}
    \sigma^0(z_i,y_i)=\begin{cases}
    \sqrt{2}\arctan\left(e^{-2y_i}\right) & z_i < -r-h\\
    \pi/\sqrt{2} & z_i = -r-h, y_i<0\\
    0 & z_i=-r-h, y_i\geq0\\
    \pi/\sqrt{2} & z_i = -r, y_i<0\\
    \sqrt{2}\pi & z_i=-r, y_i\geq0\\
    -\sqrt{2}\arctan\left(e^{-2y_i}\right)+\sqrt{2}\pi & -r < z_i < r\\
    \pi/\sqrt{2} & z_i = r-h, y_i<0\\
    \sqrt{2}\pi & z_i=r-h, y_i\geq0\\
    \pi/\sqrt{2} & z_i = r, y_i<0\\
    0 & z_i=r, y_i\geq0\\
    \sqrt{2}\arctan\left(e^{-2y_i}\right) & z_i > r
    \end{cases},
\end{align}
as shown in Figure \ref{fig:deconfinement_initial}. The final result is shown in Figure \ref{fig:deconfinement_result}. The energy density was already shown in Figure~\ref{fig:02} in the Introduction.

\begin{figure}[h]
    \centering
    \includegraphics[width=0.8\linewidth]{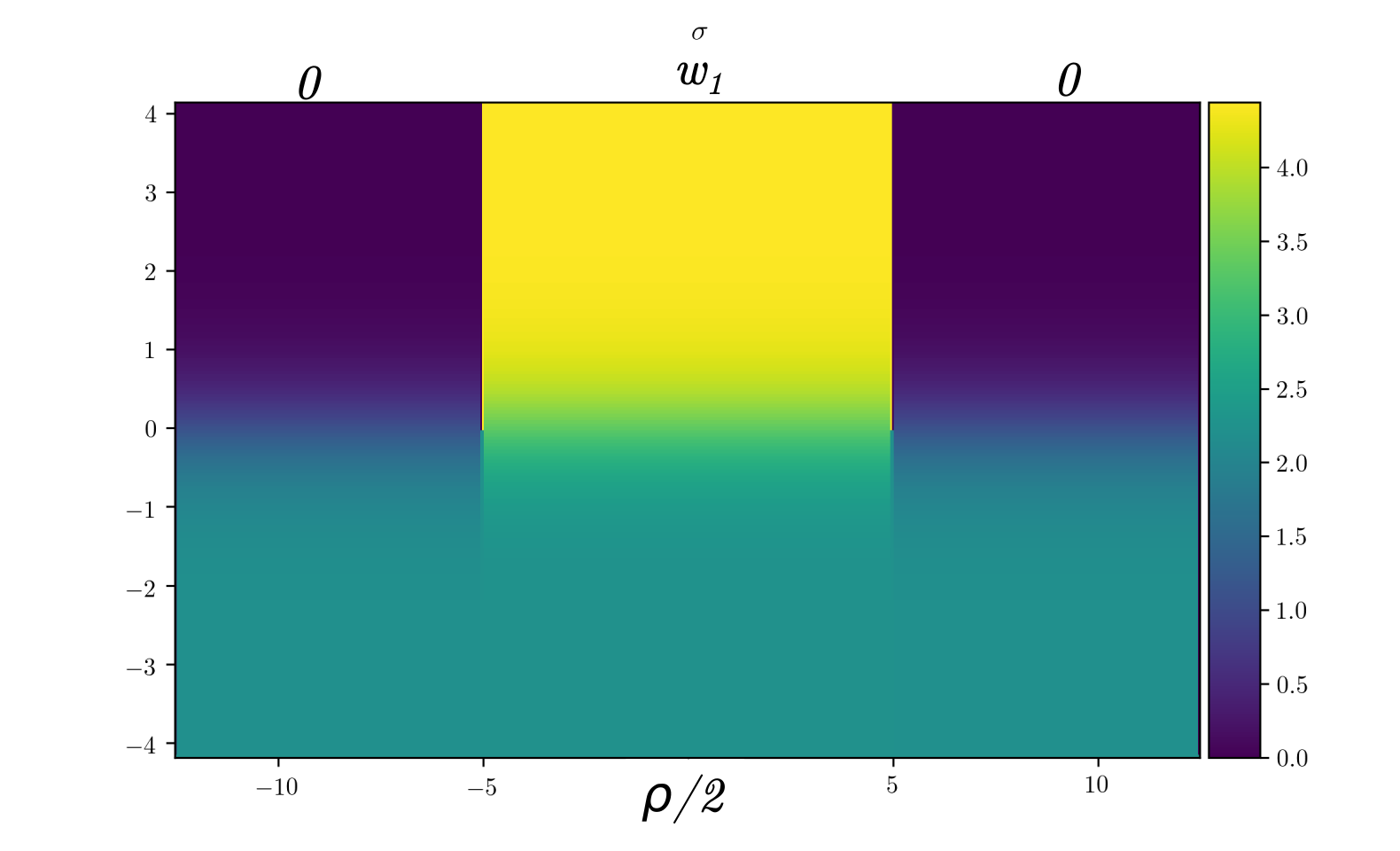}
    \caption{Initial Grid for Solving Deconfinement in $SU(2)$}
    \label{fig:deconfinement_initial}
\end{figure}

\begin{figure}[h]
    \centering
    \includegraphics[width=\linewidth]{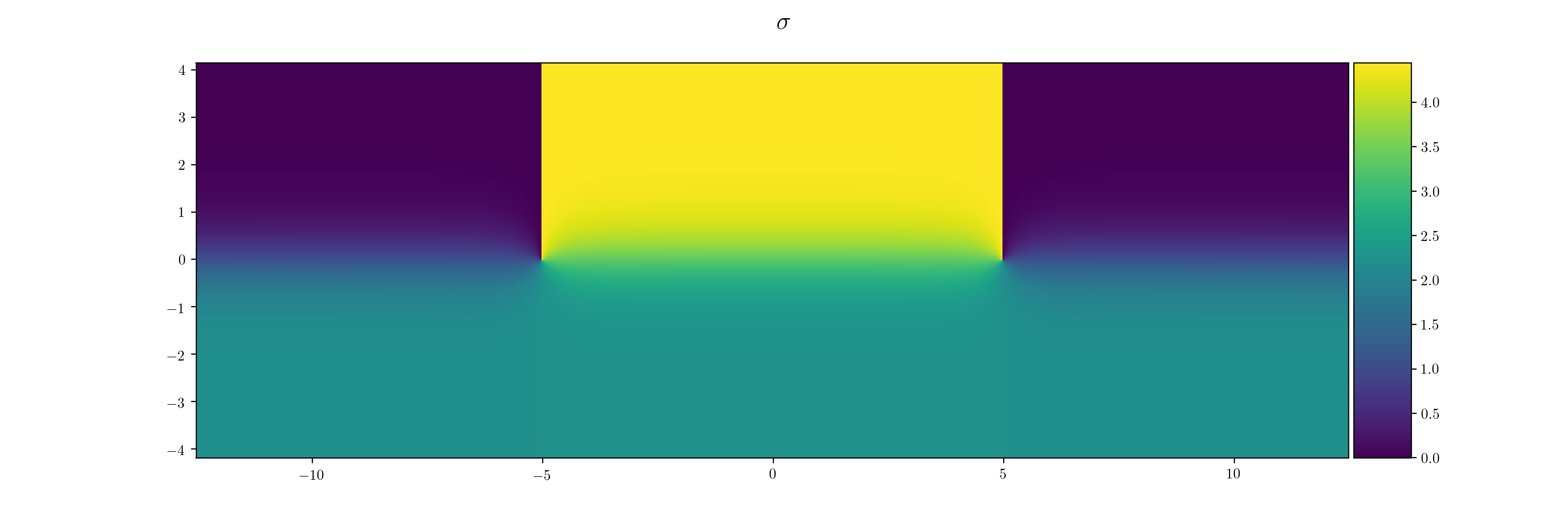}
    \caption{The dual photon solution describing   deconfinement of quarks on $SU(2)$ DWs. Grid is $600\times200$ points, with a lattice spacing of 0.42, making the grid size $25\times8.33$ (ie $L=25, H=8.33$), and the algorithm was run for 1000 iterations. The quarks are spaced 10 units apart, corresponding to 10 dual photon Compton wavelengths.}
    \label{fig:deconfinement_result}
\end{figure}

To show that the two quarks are deconfined, we solve for the fields for different quark separations, and then plot the energy of the resulting configurations as a function of the separation. If the energy is constant as a function of the spacing, then we can say that the quarks are deconfined. To test this, with a grid width of 25 units (i.e., dual-photon Compton wavelengths) we solved for configurations with spacings from 0.5 units apart all the way to 12.5 units (since we have periodic boundary conditions, 12.5 units is the farthest we could get the quarks apart). Figure \ref{fig:deconfinement_energy} shows the resulting energies relative to the energy at 12.5 unit spacing as a function of the relative separation, showing clearly that the quarks are deconfined.

\begin{figure}[h]
    \centering
    \includegraphics[width=0.7\linewidth]{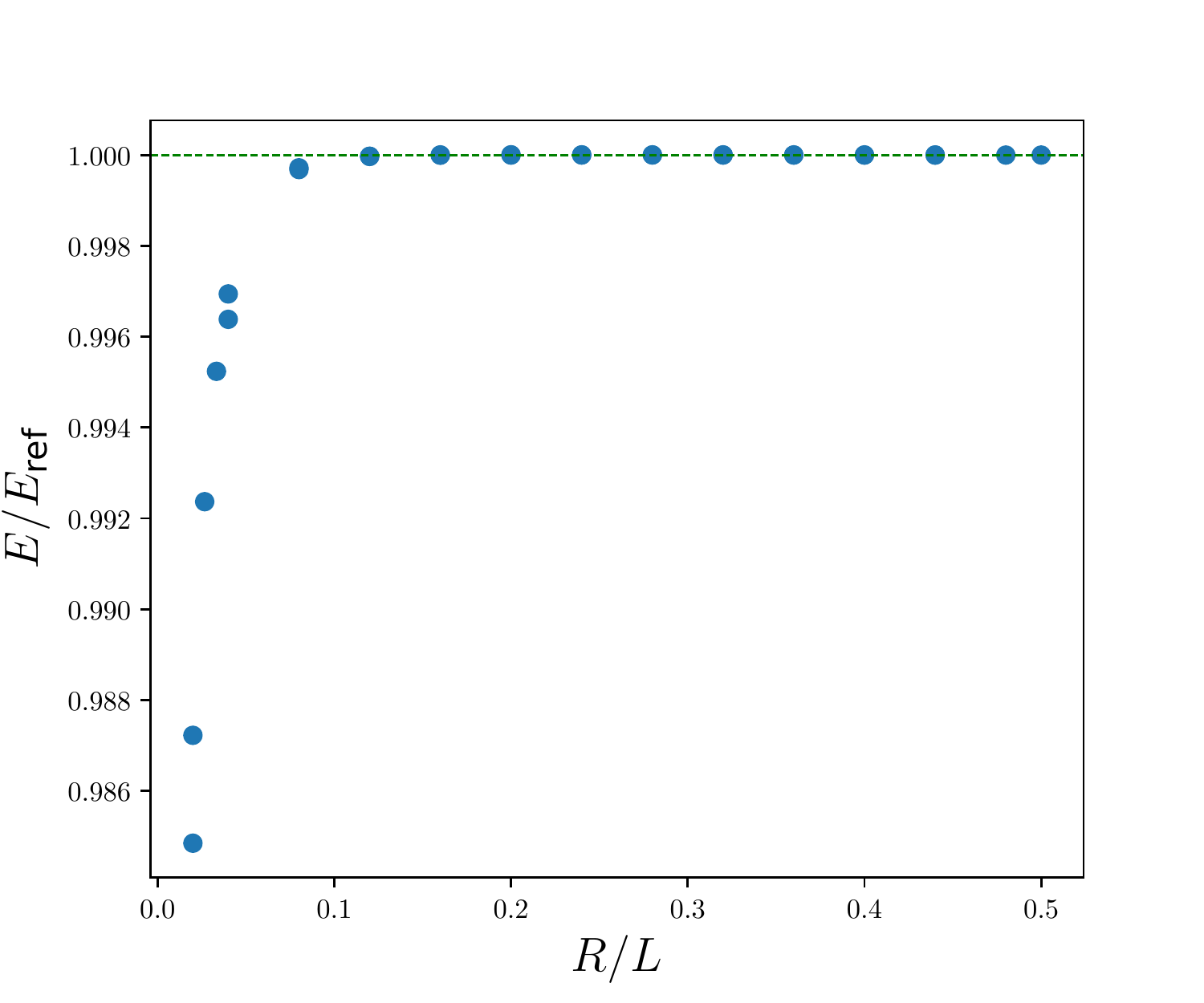}
    \caption{Energy of two-quark configurations in $SU(2)$, relative to the energy at the farthest separation, as a function of the relative separation. We can see that for separations greater than about 10\% of the grid width the energy is constant and thus the quarks are deconfined (this separation corresponds to 2.5 inverse dual photon masses).
    The dashed line at $E/E_{\text{ref}}=1$ is for reference.}
    \label{fig:deconfinement_energy}
\end{figure}

\subsection{The Static Baryon in $\mathbf{SU(3)}$: SYM vs dYM}
\paragraph{}
A baryon in $SU(3)$ consists of three quarks, and we will consider a configuration where the quarks are laid out in an upright equilateral triangle with side length $R$ centered at the origin. Thus, the quark at the top (quark 1) will be at $(z_1,y_1)=(0,R/\sqrt{3})$, the quark on the bottom right (quark 2) will be at $(z_2,y_2)=(R/2,R/2\sqrt{3})$, and the quark on the bottom left (quark 3) will be at $(z_3,y_3)=(-R/2,R/2\sqrt{3})$. We choose quark 1 to have a charge of $\vec{\nu}_1=\vec{w}_1$, quark 2 to have a charge of $\vec{\nu}_2=\vec{w}_2-\vec{w}_1$, and quark 3 to have a charge of $\vec{\nu}_3=-\vec{w}_2$. We put the jump for quark 1 extending vertically upwards from the quark, and the jumps for quarks 2 and 3 extending vertically downwards from the quarks. Thus, the equations of motion are
\begin{align}\label{eqn:eom_baryon}
    \nabla^2{x^a}=\frac{1}{4} {\partial \over \partial \left(x^*\right)^a}\left\lvert {d W \over d \vec{x}}\right\rvert^2+2\pi i&\left[\left(\vec{\nu}_1\right)_a\partial_z\delta(z-z_1)\int_{y_1}^{\infty} {d y'}\delta(y-y')+\right.\\
    &\left.\left[\left(\vec{\nu}_2\right)_a\partial_z\delta(z-z_2)+\left(\vec{\nu}_3\right)_a\partial_z\delta(z-z_3)\right]\int_{y_2=y_3}^{-\infty} {dy'}\delta(y-y')\right].\nonumber
\end{align}

\begin{figure}[h]
    \centering
    \includegraphics[width=\linewidth]{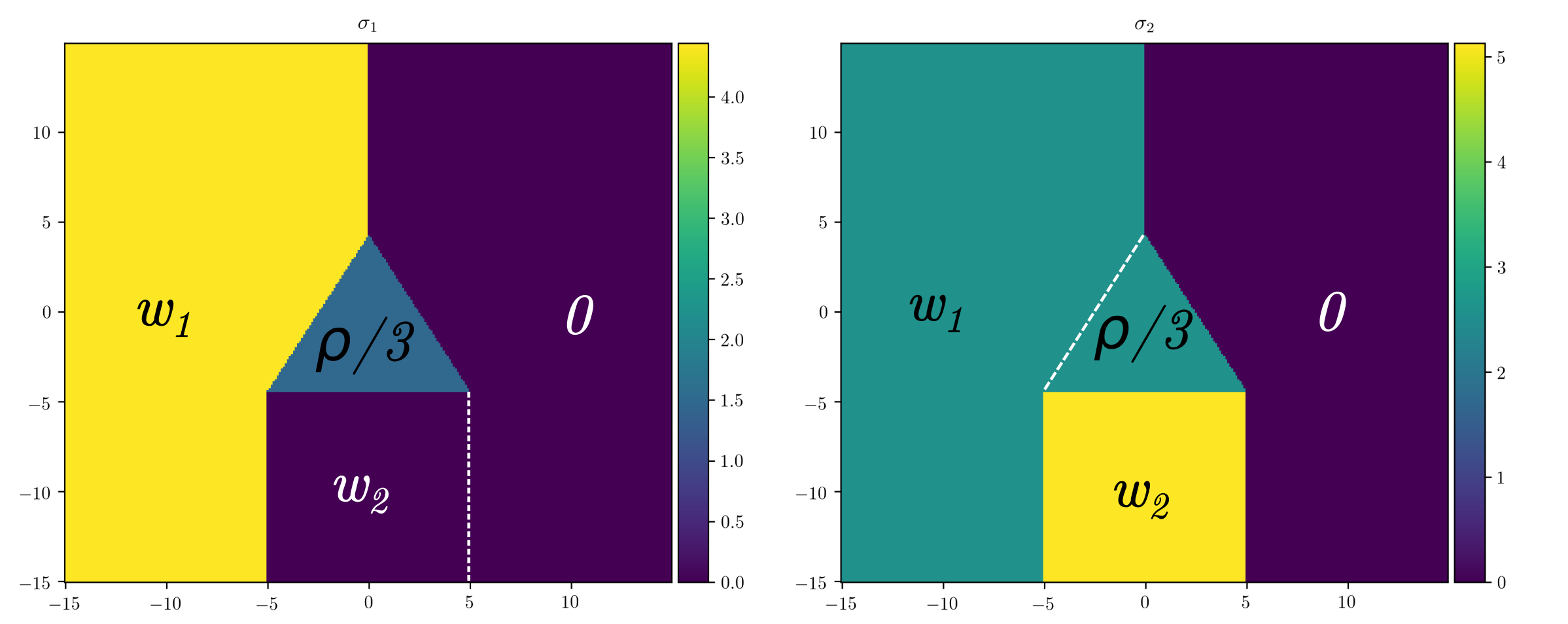}
    \caption{Initial conditions for $SU(3)$ baryon in supersymmetric Yang-Mills theory. Dashed lines simply show the boundaries between regions when it is unclear.}
    \label{fig:SU3_Baryon_initial}
\end{figure}

\begin{figure}[h]
    \centering
    \includegraphics[width=\linewidth]{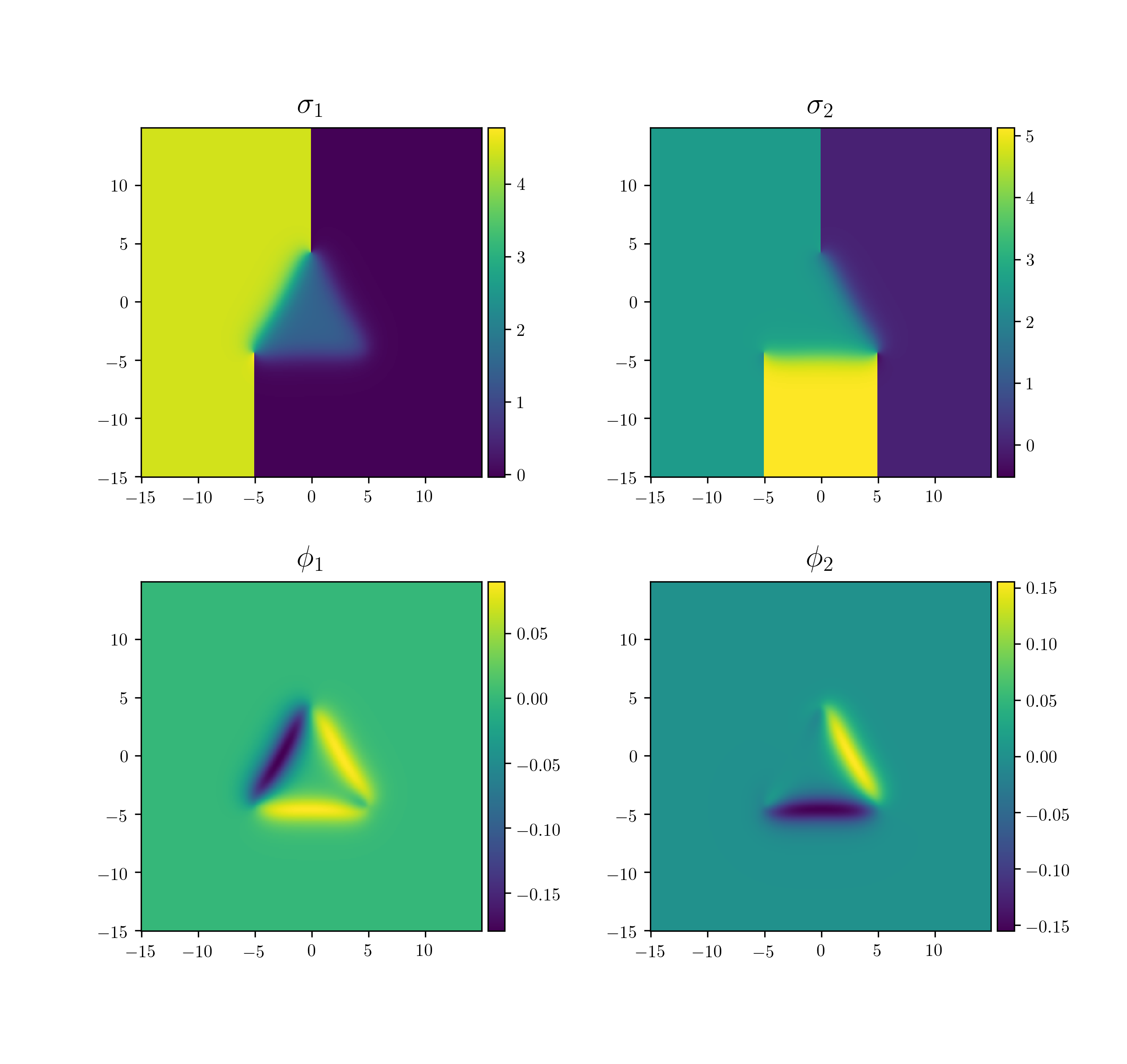}
    \caption{Results for $SU(3)$ baryon in supersymmetric Yang-Mills theory, where quarks are spaced 10 units apart. The grid is $300\times300$ points, with a grid spacing of $0.1$ so that $H=30=L$. The algorithm was run for 1000 iterations.}
    \label{fig:SU3_Baryon_result}
\end{figure}

We set the boundary conditions for $\sigma$ to be as follows: the boundary to the left of quarks 1 and 3 is $2\pi\vec{w}_1$, the boundary to the right of quarks 1 and 2 is $\vec{0}$, and the boundary below and between quarks 2 and 3 is $2\pi \vec{w}_2$. Written as an equation, with a grid of size $L\times H$, the boundaries are specified as
\begin{align}
    \sigmav(-L/2,y)&=2\pi\vec{w}_1 
    \end{align}
    \begin{align}
    \sigmav(z,H/2)&=\begin{cases}2\pi\vec{w}_1 & z<0 \\ \vec{0} & z\geq 0\end{cases} 
    \end{align}
    \begin{align}
    \sigmav(L/2,y)&=\vec{0}\end{align}
    \begin{align}
    \sigmav(z,-H/2)&=\begin{cases}
    \vec{0} & z>\frac{R}{2} \\
    2\pi\vec{w}_2 & -\frac{R}{2}\leq z\leq \frac{R}{2} \\
    2\pi\vec{w}_1 & z<\frac{R}{2}
    \end{cases},
\end{align}
with $\phiv=\vec{0}$ everywhere.

\begin{figure}[h]
    \centering
    \includegraphics[width=\linewidth]{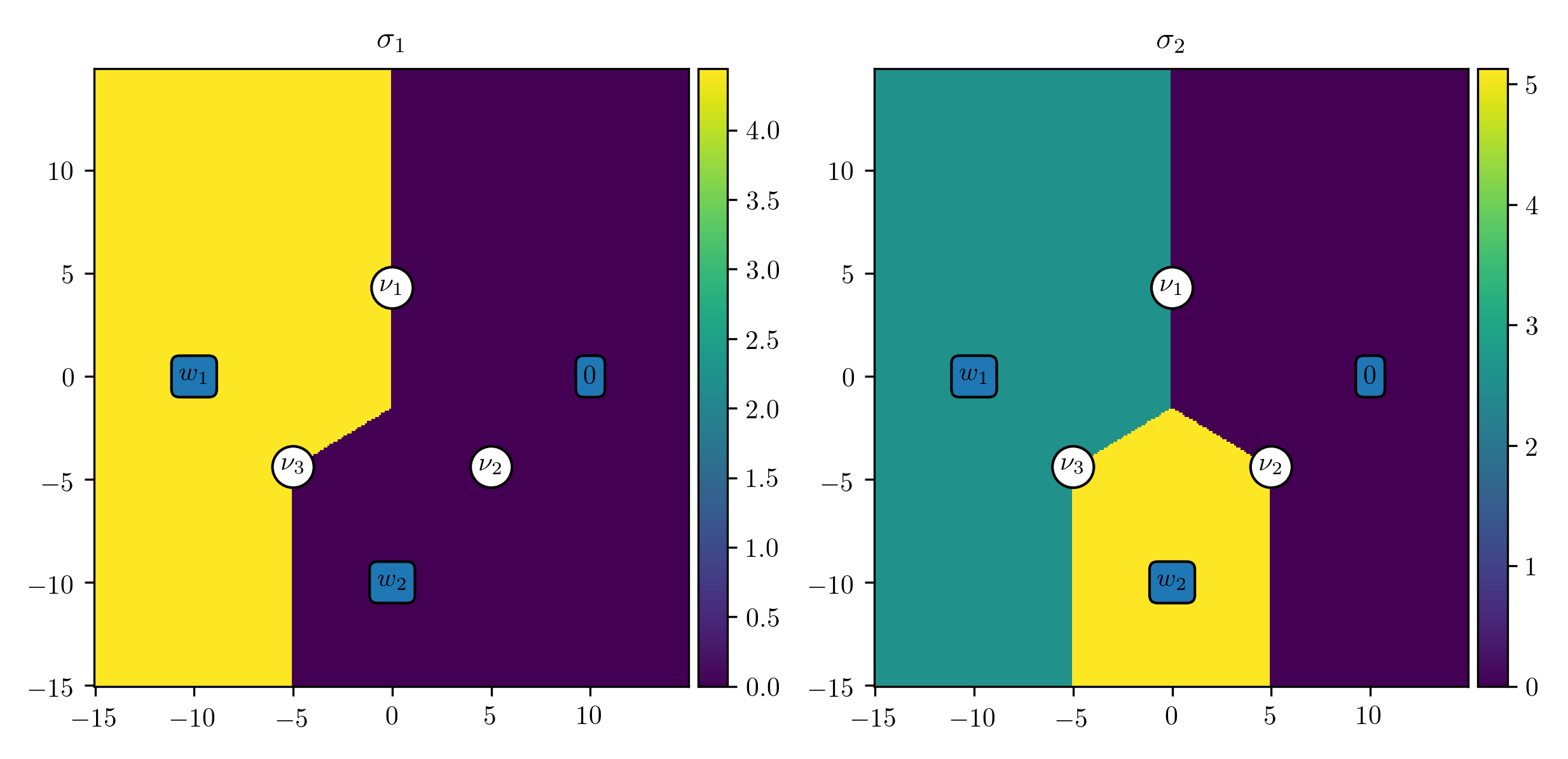}
    \caption{Initial conditions for $SU(3)$ deformed Yang-Mills baryon. For clarity, quarks are drawn as white circles, labelled by their charges, while the boundaries are labelled by the text in blue boxes.}
    \label{fig:SU3_Baryon_dYM_initial}
\end{figure}

\begin{figure}[h]
    \centering
    \includegraphics[width=\linewidth]{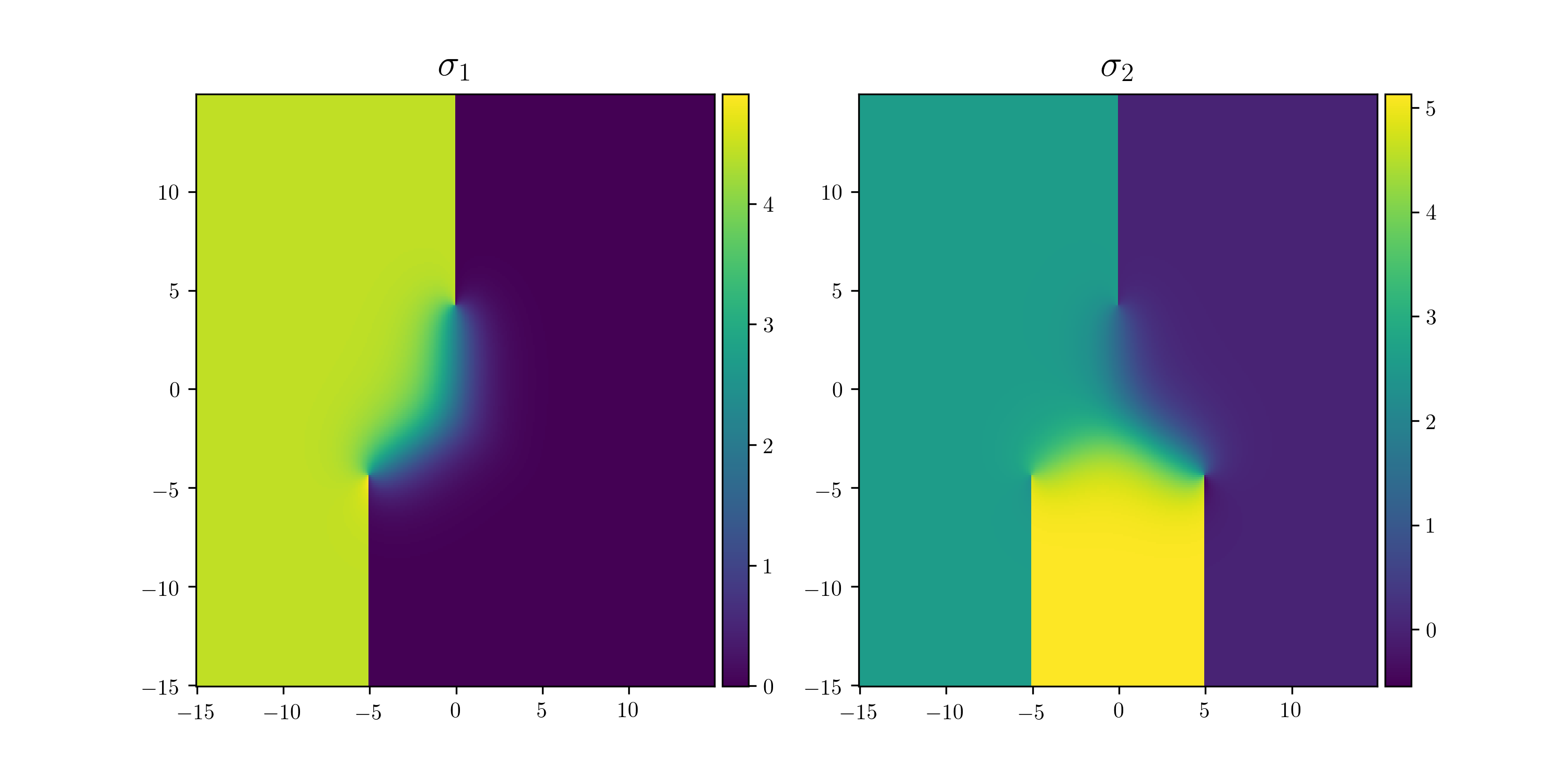}
    \caption{Results for $SU(3)$ baryon in deformed Yang-Mills theory, with quarks spaced 10 units apart. The grid is $300\times300$ points, with a grid spacing of 0.1, making $H=30=L$. The algorithm was run for 1000 iterations.}
    \label{fig:SU3_Baryon_dYM_result}
\end{figure}

\subsubsection{Supersymmetric Yang-Mills Theory}
The superpotential here is the same as used everywhere else in the paper, the inclusion of supersymmetric in the title is simply to distinguish this from the deformed Yang-Mills baryon discussed later.

We expect that the area inside the triangle formed by the quarks will go to the $k=1$ vacuum, $\sigmav=2\pi\vec{\rho}/3$, so we set the initial conditions such that inside the triangle $\sigmav=2\pi\vec{\rho}/3$. Outside the triangle, we set the initial condition to match the corresponding boundary condition. As with the boundary conditions, we set $\phiv=\vec{0}$ everywhere initially. To get a better idea of the initial conditions, see Figure \ref{fig:SU3_Baryon_initial}. See Figure \ref{fig:SU3_Baryon_result} for the final result.

\subsubsection{Deformed Yang-Mills Theory}
\paragraph{}
This nonsupersymmetric theory \cite{Unsal:2008ch} is not described by a superpotential, but we can use the  equation of motion of SYM,  (\ref{eqn:eom_2d}), simply by replacing $\left\lvert {d W_{ } \over d \vec{x}}\right\rvert^{2}$ there by 
\begin{equation}
    \left\lvert {d W_{ {\rm dYM}} \over d \vec{x}}\right\rvert^{2}=\sum_{a=1}^{N}\left(1-\cos\left(\alphav_a\cdot\sigmav\right)\right).
\end{equation}
Notice that deformed Yang-Mills has no dependence on $\phiv$, that is, we ignore $\phiv$. The boundary conditions are the same as with supersymmetric Yang-Mills, and the equations of motion are still given by equation \eqref{eqn:eom_baryon} with the above replacement.
 
We expect the three quarks to form domain walls in a ``Mercedes star" configuration, so we set the initial conditions to reflect that: we use the boundary values to ``fill" in up to the closest side of the star. See Figure \ref{fig:SU3_Baryon_dYM_initial} for clarity. Finally, see Figure \ref{fig:SU3_Baryon_dYM_result} for the final result as well as Figure \ref{fig:03} from the Introduction.

\bigskip

{\flushleft{\bf Acknowledgements:}} We thank Mohamed Anber and Francesco Benini for discussions. This work is supported by an NSERC Discovery Grant.

\appendix

\section{Derivation of results regarding magnetless solitons}
\label{magnetless}
\begin{singlespace}
  We call a soliton that has no magnetic field a magnetless soliton. It is a special kind of BPS soliton. The following is a summary of our study of this kind of solutions:
  \begin{enumerate}
\item We first show that a magnetless soliton is only possible for even $N$ and for $k=\frac{N}{2}$ walls. This basic result will be assumed in all later parts.
\item Then we prove a special case of the $ {N \choose k} = {N-1 \choose k-1} + {N-1 \choose k}$ hypothesis, applied to magnetless solitons. This is included because this result is needed for later parts.
\item Next, we define a sequence of $N$ numbers called the ``d-sequence" that completely characterize a given magnetless soliton. We prove a series of four interrelated lemmas about the d-sequence.
Finally, we use the four lemmas to prove the key theorem, which shows that there are only at most 6 possible boundary conditions that  can give rise to magnetless solitons, and derive which 6 boundary conditions these are.
\end{enumerate}
  {  
    \begin{theorem} \label{thm:N/2}
    A magnetless soliton can only occur for even $N$ and for $k=\dfrac{N}{2}$ wall.
\end{theorem}
\begin{proof}
Without loss of generality, we can assume that a $k$-wall starts from a vacuum equivalent to the origin and ends on one of the other vacua in the unit cell. All other $k$-walls are equivalent to one of these solutions up to a sign and a constant vector. We will assume working with one of these walls from now on. Then the BPS equation is of the form
\begin{equation}\label{dif1}
    \begin{cases}
        \frac{d\vec{x}}{dz} = \frac{\alpha}{2} \sum_{a=1}^{N} e^{\vec{\alpha_a} \cdot \vec{x}^*} \vec{\alpha_a}  \\
        \alpha = e^{i \theta} = \frac{W(\vec{x}(\infty)) - W(\vec{x}(- \infty))}{|W(\vec{x}(\infty)) - W(\vec{x}(- \infty))|} \\
        \vec{x}(- \infty) = \cjwj; \quad c_j = 0,1 \\
        \vec{x}(\infty) = i2\pi \frac{k}{N} \vec{\rho}
    \end{cases}
\end{equation}
Note that the boundary condition is described by $c_a$, a series of $N-1$ numbers that are either zero or one. We will call them the ``c-sequence," as they will be mentioned a lot later. The differential equation (\ref{dif1}) can be rewritten in real- and imaginary-component form as
\begin{equation} \label{eq:first phi equation}
        {d \vec\phi \over d z} = \frac{1}{2} \sumaN e^{\ala \cdot \vec{\phi}} \cos(\theta - \ala \cdot \vs) \ala
~,~~    {d \vec\sigma \over d z} = \dfrac{1}{2} \sumaN e^{\ala \cdot \vec{\phi}} \sin(\theta - \ala \cdot \vs) \ala~.
\end{equation}
A magnetless soliton has $\vp = 0$. Then the $\vp$ equation becomes
$
    0 =  \sumaN  \cos(\theta - \ala \cdot \vs) \ala
$.
Since $ \alN = - \sum_{a=1}^{N-1} \ala$, this implies
\begin{equation}
    0 = \sum_{a=1}^{N-1} \left[ \cos(\theta - \ala \cdot \vs) - \cos(\theta - \alN \cdot \vs) \right] \ala
\end{equation}
Since $\alpha_1, \dots, \alpha_{N-1}$ are linearly independent,
\begin{equation}
    \cos(\theta - \ala \cdot \vs) = \cos(\theta - \alN \cdot \vs); \quad a=1, \dots, N-1
\end{equation}
For each $a$, there are two possible cases:

{\it Case 1: }
\begin{equation} \label{eq:case 1}
\theta - \ala \cdot \vs = 2\pi n_a - (\theta - \alN \cdot \vs)
\end{equation}

{\it Case 2: }
\begin{equation} \label{eq:case 2}
    \theta - \ala \cdot \vs = \theta - \alN \cdot \vs + 2\pi n_a 
\end{equation}
where $n_a$ are integers. We can combine the two cases into a single equation by defining $d_a = \pm 1$. For later reference, we call these numbers of one or minus one the ``d-sequence".
\begin{equation}
    \theta - \ala \cdot \vs = d_a \left( \theta - \alN \cdot \vs \right) + 2\pi n_a
\end{equation}

If $d_a = -1$ for any $a$, pick that $a$. Then we have equation (\ref{eq:case 1}) from {\it Case 1} and the equation reduces to:
\begin{equation} \label{eq:alpha square}
        \alpha^2 = e^{i (\ala + \alN) \cdot \vs}
\end{equation}
Since this is true for all $z$, we substitute $\vs(-\infty) = 2\pi \sum_{j=1}^{N-1} c_j \vec{w}_j$ into equation (\ref{eq:alpha square}). Using the identity $ \ala \cdot \vec{w}_j = \delta_{aj} - \delta_{aN} $, all of the terms in the exponential must be integer multiple of $i 2 \pi$, and so
\begin{equation} \label{eq:sub -infinity}
    \alpha^2 = 1
\end{equation}
Next, substitute instead $\vs(\infty) = 2\pi \dfrac{k}{N} \vec{\rho}$ into equation (\ref{eq:alpha square}), and using the identity $ \ala \cdot \vec{\rho} = 1 - N \delta_{aN}$, we have
\begin{equation} \label{eq:sub +infinity}
    \alpha^2 = \exp \left( i 2\pi \cdot \dfrac{2k}{N} \right)~.
\end{equation}
Equating equation (\ref{eq:sub -infinity}) and (\ref{eq:sub +infinity}), we have that $
     1 = \exp \left[ i 2\pi \cdot \dfrac{2k}{N} \right]$. Since $ 0 < k \leq \dfrac{N}{2} $ and $k$ must be an integer, we must have that $     k = \dfrac{N}{2}
$
and $N$ must be even.

So, we have proven the theorem if at least one $d_a = -1$. The only case left to be considered is if $d_a = 1$ for all $a = 1, \dots, N-1$. Then equation (\ref{eq:case 2}) is true for every $a$, and we can rewrite it as:
\begin{equation} \label{eq:all da=1}
    (\alN - \ala) \cdot \vs = 2 \pi n_a
\end{equation}
Since this is true for all $z$, we substitute $\vs(-\infty) = 2\pi \sum_{j=1}^{N-1} c_j \vec{w}_j$ to get
\begin{equation} \label{eq:sub -infinity for all da=1}
 \left( \sumcj \right) + c_a = -n_a
\end{equation}
If we substitute instead $\vs(\infty) = 2\pi \dfrac{k}{N} \vec{\rho}$ into equation (\ref{eq:all da=1}) and after simplifying, we arrive at \begin{equation} \label{eq:sub +infinity for all da=1}
    k = -n_a
\end{equation}
So it turns out all of the $n_a$ are the same. Equating equation (\ref{eq:sub -infinity for all da=1}) and (\ref{eq:sub +infinity for all da=1}), we have
\begin{equation} \label{eq:k equal sum of ca}
    k =  \left( \sumcj \right) + c_a
\end{equation}
Since this must be true for all $a$, this implies that all the $c_a$ are also the same: $c_1 = c_2 =\dots =c_{N-1}$. If they are all $0$, then that implies $k=0$, which is impossible. If they are all $1$, then that implies $k=N$, which is also impossible. So assuming that $d_a = 1$ for all $a$ leads to contradiction.
\end{proof}

{\flushleft{
It is easy}} to adapt the proof of Theorem \ref{thm:N/2} to prove our $ {N \choose k} = {N-1 \choose k-1} + {N-1 \choose k}$ hypothesis for magnetless solitons.

\begin{theorem} \label{thm:N choose k}
    A magnetless soliton has  boundary conditions of either
    $$ i2\pi (\vec{w}_{a_1} + \dots + \vec{w}_{a_{N/2}}) \to \vec{x}_{N/2} $$
    or
    $$ i2\pi (\vec{w}_{a_1} + \dots + \vec{w}_{a_{N/2 - 1}}) \to \vec{x}_{N/2} $$
\end{theorem}
\begin{proof}
    In the proof of Theorem \ref{thm:N/2}, we have shown that for a magnetless soliton, we can always pick an $a=1,\dots,N-1$ such that $d_a = -1$ and 
    $ 2\theta = 2 \pi n_a + (\ala + \alN) \cdot \vs  $.
    Evaluating at $\pm \infty$ and subtracting the equations, we have
$
        0 = (\ala + \alN) \cdot (\vs(\infty) - \vs(-\infty) ).
$
    After simplifying, we have    $
        c_a -\sum_{j=1}^{N-1} c_j = \dfrac{k}{N} (2 - N).
$
    Since we know that $k = \dfrac{N}{2}$ by Theorem \ref{thm:N/2}, we conclude that 
$
        c_a -\sum_{j=1}^{N-1} c_j =  1 - \dfrac{N}{2}.$
    If $c_a = 1$, then
$
        \sum_{j=1}^{N-1} c_j = \dfrac{N}{2}
$
and if $c_a = 0$, then
$
        \sum_{j=1}^{N-1} c_j = \dfrac{N}{2} - 1.
$
 Since $\sum_{j=1}^{N-1} c_j$ is just the total number of terms in the boundary condition, we have completed the proof.
\end{proof}

\begin{lemma} \label{lm:d-sequence 3 conditions}
The d-sequence of a magnetless soliton must satisfy three conditions:
$$ \sumaN d_a = 0~, $$
$$d_a = \frac{\ala \cdot \Delta \vs}{ \alN \cdot \Delta \vs}, ~ \Delta \vs \equiv \vs(\infty) - \vs(-\infty)~,$$
$$ d_{N-1} + d_1 = \frac{d_{a+1} + d_{a-1}}{d_a}~, $$
for $a = 1, \dots, N$. The second equation can be taken to be the definition of the d-sequence. In the third equation, we make the identification $0 \equiv N$.
\end{lemma}
\begin{proof}
As an obvious consequence of Theorem \ref{thm:N/2}, all magnetless solitons are equivalent (up to a sign and a constant) to a solution whose projection into the $W$-plane is a straight line lying on the real axis, going from $W(\xninf) = N$ to $W(\xinf)=-N$. In this case, it is clear that
$
    \alpha = -1.
$
For such a magnetless soliton, the BPS equations  (\ref{eq:first phi equation})  become
\begin{equation} \label{eq:phi}
    0 = \sum_{a=1}^N \ca \ala ~,
  ~~  \frac{d \vs}{dz} = \frac{1}{2} \sum_{a=1}^N \sa \ala ~.
\end{equation}
In addition, the fact that the W-line lies on the real axis is given by the condition
\begin{equation}
    \text{Im}(W(\vec{x}))  = 0  
    \implies \sum_{a=1}^N \sa  = 0 ~. \label{eq:horizontal}
\end{equation}
In the proof of Theorem \ref{thm:N/2}, we saw that Equation (\ref{eq:phi}) implies:
\begin{equation} \label{eq:cos equal}
    \ca = \cN; \quad a = 1, \dots, N-1
\end{equation}
Equations (\ref{eq:horizontal}) and (\ref{eq:cos equal}) together give a constraint on the $d$-sequence. To see this, we first let $\sN = S_N$. For any $a = 1, \dots, N-1$, equation (\ref{eq:cos equal}) implies that
$$\ala \cdot \vs = \alN \cdot \vs + 2\pi n_a
\quad \text{or} \quad 
\ala \cdot \vs = 2\pi n_a - \alN \cdot \vs   $$
$$ \implies \sa = S_N \quad \text{or} \quad \sa= -S_N  $$
\begin{equation} \label{eq:sin(a) sin(N)}
   \implies  \sa = d_a S_N
\end{equation}
where $d_a = \pm 1$ are the d-sequence we defined before. Substituting equation (\ref{eq:sin(a) sin(N)}) into equation (\ref{eq:horizontal}), $S_N$ cancels out and we get our first condition, 
$
    \sumaN d_a = 0,
$
where we have extended the definition of $d_a$ such that
$
    d_N =1
$.

Next, equation (\ref{eq:cos equal}) implies that
\begin{equation} \label{eq: alpha dot sigma}
    \left( \ala- d_a \alN \right) \cdot \vs = 2\pi n_a 
\end{equation}
Since this must be true for all $z$, they must also be true at $\pm \infty$. Let $\Delta \vs = \vs(\infty) - \vs(-\infty)$. Subtracting the equation evaluated at $\pm \infty$, we have
\begin{equation} \label{eq:da formula}
   \left( \ala - d_a \alN \right) \cdot \Delta \vs = 0    \implies d_a = \frac{\ala \cdot \Delta \vs}{ \alN \cdot \Delta \vs}; \quad a = 1, \dots, N
\end{equation}
which yields an explicit formula for $d_a$, which is our second condition.

Next, we differentiate equation (\ref{eq: alpha dot sigma}) by $z$ to get
$$ \left( \ala- d_a \alN \right) \cdot \frac{d \vs}{dz} = 0 $$
Substituting equation (\ref{eq:phi}) and equation (\ref{eq:sin(a) sin(N)}), we get
\begin{equation}
    \sumbN d_b \left[ \ala \cdot \alb - d_a \alN \cdot \alb \right]  = 0
\end{equation}
Applying the identity $\ala \cdot \alb = 2 \delta^{ab} - \delta^{a,b+1} - \delta^{a+1,b}$, where $0 \equiv N$, we have $
    2d_a - d_{a-1} - d_{a+1} - 2d_a d_N + d_a d_{N-1} + d_a d_1 = 0
$. Using $d_N = 1$, this simplifies to the third desired relation:
$
    d_{N-1} + d_1 = \frac{d_{a+1} + d_{a-1}}{d_a}
$, completing the proof.
\end{proof}

\begin{lemma} \label{lm:HMS Algorithm}
There are the three distinct kinds of magnetless solitons, characterized by their d-sequence.
\begin{enumerate}
    \item Alternating: $$d_1 = -1, d_2 = 1, d_3 =-1, \dots, d_{N-1} = -1, d_N = 1$$
    \item Positive Paired: $$ d_1 = 1, d_2 = -1, d_3 = -1, d_4 = 1, d_5 = 1, \dots, d_N = 1 $$
    \item Negative Paired: $$ d_1 = -1, d_2 = -1, d_3 = 1, d_4 = 1, \dots, d_N = 1 $$
\end{enumerate}
\end{lemma}
\begin{proof}
There are several ways to realize the third condition of Lemma \ref{lm:d-sequence 3 conditions}.

\textbf{Case 1}: $d_{N-1} = 1$ and $d_1 = -1$.
It is clear that the only way this is possible is if
$$ d_0 = d_N = 1, d_1 = -1, d_2 = -1, d_3 = 1, d_4 = 1, \dots, d_{N-1} = 1, d_N = 1. $$

\textbf{Case 2}: $d_{N-1} = -1$ and $d_1 = 1$.
This requires
$$ d_0 = d_N=1, d_1 = 1, d_2 = -1, d_3 = -1, d_4 = 1, \dots, d_{N-1} = -1, d_N = 1.$$

\textbf{Case 3}: $d_{N-1} = 1$ and $d_1 = 1$. Then
$  \frac{d_{a+1} + d_{a-1}}{d_a} =2,  $
which must be true for all $a$. This implies that either
$$ d_{a+1} = 1, d_{a-1} = 1, d_a = 1 ~~
{\rm or}~~
  d_{a+1} = -1, d_{a-1} = -1, d_a = -1,$$
both of which lead to the contradiction that $ \sumaN d_a \neq 0$, forbidden by first condition of Lemma \ref{lm:d-sequence 3 conditions}. So case 3 is impossible.

\textbf{Case 4}: $d_{N-1} = -1$ and $d_1 = -1$.
This implies
$$ \dots, d_{N-2} = 1, d_{N-1} = -1, d_N = 1, d_1 = -1, \dots$$
So only case 1, 2, and 4 are possible,  corresponding to our three types of magnetless solitons.
\end{proof}

One remark is that any kink and anti-kink belong to the same type of magnetless soliton, since the sign of $\Delta \vs$ cancels out in the condition 2 of Lemma \ref{lm:d-sequence 3 conditions}.

\begin{lemma} \label{thm:N mod 4}
The paired magnetless solitons (both positive and negative) can only exist if 
$$N \mod 4 =0$$
while alternating magnetless solitons can exist for all even $N$.
\end{lemma}
\begin{proof}
    Notice that if $N=6$, then a positive paired magnetless soliton has
$$ d_1 = 1, d_2 = -1, d_3 = -1, d_4 =1, d_5 =1, d_6 = -1 $$
But this contradicts $d_6 = d_N = 1$. Similarly, a negative paired magnetless soliton has
$$ d_1 = -1, d_2 = -1, d_3 =1, d_4 =1 , d_5 = -1, d_6 = -1 $$
The generalization to $N \mod 4 \neq 0$ is obvious.
\end{proof}

We now combine Lemma \ref{lm:HMS Algorithm} with the proof of Theorem \ref{thm:N choose k} to arrive at a ``dictionary" between the d-sequence and the c-sequence, in the sense that given a certain condition, knowing the value of $d_a$ immediately tells us the value of $c_a$.

\begin{lemma} \label{thm:dictionary}
    Consider a magnetless soliton of the form:
    $$ \cjwj \to \vec{x}_{N/2}. $$
    
    If \; $\sum_{j=1}^{N-1} c_j = \frac{N}{2}$, then 
    $\begin{cases}
    d_a = -1 \implies c_a = 1 \\
    d_a = 1 \implies c_a = 0 \\
    \end{cases}$
    
    If \;$\sum_{j=1}^{N-1} c_j = \frac{N}{2} - 1$, then  
    $\begin{cases}
    d_a = -1 \implies c_a = 0 \\
    d_a = 1 \implies c_a = 1 \\
    \end{cases}$
\end{lemma}
\begin{proof}
    In the proof of Theorem \ref{thm:N choose k}, we chose an $a$ such that $d_a = -1$ (which is always possible). Then we showed that, if $c_a = 1$, then $
        \sum_{j=1}^{N-1} c_j = \dfrac{N}{2}$, and  if $c_a = 0$, then
$
        \sum_{j=1}^{N-1} c_j = \dfrac{N}{2} - 1
$.
    We can rewrite this result as
    $$ d_a = -1 \implies \left[ \left( c_a = 1 \implies \sumcj = \frac{N}{2} \right) \text{ and } \left( c_a = 0 \implies \sumcj = \frac{N}{2} -1 \right) \right]~.$$
    Using the fact that $c_a$ can only be 0 or 1, and the fact that $\sumcj$ can only be $\dfrac{N}{2}$ or $\dfrac{N}{2} -1 $, we can reverse the direction of implication inside the square bracket:
    $$ d_a = -1 \implies \left[ \left( c_a = 1 \iff \sumcj = \frac{N}{2} \right) \text{ and } \left( c_a = 0 \iff \sumcj = \frac{N}{2} -1 \right) \right]~.$$
    Now suppose $\sumcj = \frac{N}{2}$. We have $d_a = -1 \implies c_a = 1 $. Similarly, suppose $\sumcj = \frac{N}{2} - 1$. We have $d_a = -1 \implies c_a = 0 $. This completes half of the dictionary in both cases. 
    
    We now want to know what $d_a = 1$ implies in either case.
    So suppose $d_a = 1$. From the proof of Theorem \ref{thm:N/2}, we know that this implies equation (\ref{eq:k equal sum of ca}):
$
    k =  \left( \sumcj \right) + c_a~.
$
    We also know from Theorem \ref{thm:N/2} that $k = \dfrac{N}{2}$, hence
    $  c_a + \sumcj = \frac{N}{2} $.
    Finally, we see that if $\sumcj = \dfrac{N}{2}$, then
    $ d_a = 1 \implies c_a = 0, $
    and if $\sumcj = \dfrac{N}{2} - 1 $, then
    $  d_a = 1 \implies c_a = 1. $
 Combining with what we showed before, we have a dictionary that tells us what $c_a$ is given $d_a$.
\end{proof}
Next, we put everything together to determine all the possible boundary conditions that can give rise to magnetless solitons in $SU(N)$.
\begin{theorem}\label{theorem3}
    Let $N$ be even. If $N \mod 4 \neq 0$, then there are at most $2$ magnetless solitons, with boundary conditions
    $$ i2 \pi \left( \underbrace{ \vec{w}_1 + \vec{w}_3 + \dots + \vec{w}_{N-1} }_{\text{all odd terms}} \right) \to \vec{x}_{\frac{N}{2}}$$
    $$ i2 \pi \left( \underbrace{ \vec{w}_2 + \vec{w}_4 + \dots + \vec{w}_{N-2} }_{\text{all even terms}} \right) \to \vec{x}_{\frac{N}{2}}$$
    
    If $N \mod 4 = 0$, there are at most 6 magnetless solitons, which include the above 2, as well as solutions with boundary conditions
    $$ i2\pi \left( \vec{w}_2 + \vec{w}_3 + \vec{w}_6 + \vec{w}_7 + \dots + \vec{w}_{N-2} + \vec{w}_{N-1} \right)  \to \vec{x}_{\frac{N}{2}} $$
    
    $$ i2\pi \left( \vec{w}_1 + \vec{w}_4 + \vec{w}_5 + \vec{w}_8 + \vec{w}_9 + \dots + \vec{w}_{N-4} + \vec{w}_{N-3} \right)  \to \vec{x}_{\frac{N}{2}} $$
    
    $$ i2\pi \left( \vec{w}_1 + \vec{w}_2 + \vec{w}_5 + \vec{w}_6  + \dots + \vec{w}_{N-3} + \vec{w}_{N-2} \right)  \to \vec{x}_{\frac{N}{2}} $$
    
    $$ i2\pi \left( \vec{w}_3 + \vec{w}_4 + \vec{w}_7 + \vec{w}_8  + \dots + \vec{w}_{N-5} + \vec{w}_{N-4} + \vec{w}_{N-1} \right)  \to \vec{x}_{\frac{N}{2}} $$
\end{theorem}
\begin{proof}
    By Lemma \ref{thm:N mod 4}, of the 3 types of magnetless solutions (described in Lemma \ref{lm:HMS Algorithm}), the alternating type can exist for all even N, while the paired type (both positive and negative) can only exist if $N \mod 4 = 0$.
    
    \textcircled{1} Suppose $\sumcj = \frac{N}{2}$. An alternating soliton has the following sequence of $d_a$:
    $$ d_1 = -1, d_2=1, d_3=-1, \cdots, d_{N-2}=1,d_{N-1}=-1,d_N=1 $$
    According to the dictionary in Lemma \ref{thm:dictionary}, this translates to the following sequence of $c_a$:
    $$ c_1 = 1, c_2=0, c_3=1, \cdots, c_{N-2}=0,c_{N-1}=1 $$
    Note that there are $N$ of the $d_a$ but only $N-1$ of the $c_a$. The $d_N$ is not translated to anything.
    The boundary conditions given by this sequence is
    $$ \quad i2 \pi \left( \underbrace{ \vec{w}_1 + \vec{w}_3 + \dots + \vec{w}_{N-1} }_{\text{all odd terms}} \right) \to \vec{x}_{\frac{N}{2}}$$
    
    \textcircled{2} Now suppose $\sumcj = \frac{N}{2} - 1$. The same sequence of $d_a$ now translates to the opposite:
    $$ c_1 = 0, c_2=1, c_3=0, \cdots, c_{N-2}=1,c_{N-1}=0 $$
    which gives the boundary condition:
    $$ i2 \pi \left( \underbrace{ \vec{w}_2 + \vec{w}_4 + \dots + \vec{w}_{N-2} }_{\text{all even terms}} \right) \to \vec{x}_{\frac{N}{2}}$$
    These are present for all even $N$.
    
    If $N \mod 4 = 0$, then we have, in addition, 4 more possible boundary conditions.
    
    \textcircled{3} Suppose $\sumcj = \frac{N}{2}$. The positive paired d-sequence is 
    $$ d_1 = 1, d_2 = -1, d_3 = -1, d_4 = 1, d_5 = 1, \dots, d_{N-2}=-1,d_{N-1}=-1,d_N = 1 $$
    which translates to
    $$ c_1 = 0, c_2 = 1, c_3 =1, c_4 = 0, c_5 = 0, \dots, c_{N-2} = 1,c_{N-1} =1 $$
    This gives the bound
     $$ i2\pi \left( \vec{w}_2 + \vec{w}_3 + \vec{w}_6 + \vec{w}_7 + \dots + \vec{w}_{N-2} + \vec{w}_{N-1} \right)  \to \vec{x}_{\frac{N}{2}} $$
     
    \textcircled{4} Suppose $\sumcj = \frac{N}{2}-1$. The same d-sequence translates to
    $$ c_1 = 1, c_2 = 0, c_3 =0, c_4 = 1, c_5 = 1, \dots, c_{N-2} = 0,c_{N-1} =0 $$
    This gives the boundary condition
    $$ i2\pi \left( \vec{w}_1 + \vec{w}_4 + \vec{w}_5 + \vec{w}_8 + \vec{w}_9 + \dots + \vec{w}_{N-4} + \vec{w}_{N-3} \right)  \to \vec{x}_{\frac{N}{2}} $$
     
    \textcircled{5} Suppose $\sumcj = \frac{N}{2}$. The negative paired d-sequence is 
    $$ d_1 = -1, d_2 = -1, d_3 = 1, d_4 = 1, \dots, d_{N-3}=-1,d_{N-2}=-1,d_{N-1}=1,d_N = 1 $$
    which translates to
    $$ c_1 = 1, c_2 = 1, c_3 =0, c_4 = 0, \dots, c_{N-3} = 1,c_{N-2} =1,c_{N-1}=0 $$
    This gives the boundary condition
    $$ i2\pi \left( \vec{w}_1 + \vec{w}_2 + \vec{w}_5 + \vec{w}_6  + \dots + \vec{w}_{N-3} + \vec{w}_{N-2} \right)  \to \vec{x}_{\frac{N}{2}} $$
    
    \textcircled{6} Suppose $\sumcj = \frac{N}{2}-1$. The same d-sequence translates to
    $$ c_1 = 0, c_2 = 0, c_3 =1, c_4 = 1, \dots, c_{N-3} = 0,c_{N-2} =0,c_{N-1}=1$$
    This gives the boundary condition
    $$ i2\pi \left( \vec{w}_3 + \vec{w}_4 + \vec{w}_7 + \vec{w}_8  + \dots + \vec{w}_{N-5} + \vec{w}_{N-4} + \vec{w}_{N-1} \right)  \to \vec{x}_{\frac{N}{2}} $$
\end{proof}
{\flushleft{
Our}} final remark is that the 2 magnetless solitons that are present in all cases form an orbit under the $\mathbb{Z}_N$ $0$-form center symmetry. The other 4 magnetless solitons also form an orbit of size 4 under this symmetry. This makes sense since the linear operator $\cal{P}$ from (\ref{zeroform0}) takes a purely imaginary vector to a purely imaginary vector: ${\cal{P}}(0+i\vs) =0 + i {\cal{P}}(\vs)$. 
}
\subsection{Explicitly solving for the magnetless solutions}\label{magnetlesssolutions}

It turns out that the magnetless solutions listed in Theorem \ref{theorem3} can always be found explicitly in terms of a single function, which is essentially the DW solution for an $SU(2)$ gauge group. We shall not pursue this in all generality and will not derive this fact, but instead give an example. 
It is particularly helpful to work in coordinates rectifying the fundamental domain 
\begin{equation}
\tilde\sigma^a \equiv \vec{\sigma} \cdot \vec\alpha_a~, ~a=1,\ldots N-1.
\end{equation}
In these $\tilde\sigma^a$ coordinates the fundamental domain is mapped into a cube of sides  of length  $2\pi$ (s.t.~each of the $\tilde\sigma^a$ varies from $0$ to $2\pi$). The boundary conditions $\vec\sigma(-\infty) = 2 \pi \sum\limits_{a=1}^{N-1} c^a \vec{w}_a$ get mapped into $\tilde\sigma^a(-\infty) = 2 \pi c^a$ for $a = 1,\ldots N-1$, while  $\tilde\sigma^a(+\infty)=\pi$ for the magnetless $k=N/2$ solutions. 

For example, consider the solution of the BPS equations ((\ref{dif1}) with $\vec\phi=0$ and $k=N/2$) interpolating, in $\vec\sigma$ coordinates, between $$ \quad i2 \pi \left( \underbrace{ \vec{w}_1 + \vec{w}_3 + \dots + \vec{w}_{N-1} }_{\text{all odd terms}} \right) \to \vec{x}_{\frac{N}{2}}~$$ This is the magnetless solution with boundary conditions denoted by   \textcircled{1} in Theorem \ref{theorem3}. This magnetless BPS solution  can be explicitly written as 
\begin{eqnarray}
\tilde\sigma^a(z) &=& \tilde\sigma^{N-1}(z)~, ~~\qquad {\rm for} ~ a~ {\rm odd}, \nonumber\\
\tilde\sigma^a(z) &=& 2 \pi - \tilde\sigma^{N-1}(z)~ ~ ~{\rm for} ~ a~ {\rm even},  \\
~{\rm where}~\tilde\sigma^{N-1}(z) &=& 2\pi - 2 {\rm Arccot} \;e^{- 2(z - z_0)}.\nonumber
\end{eqnarray}
We leave it as an exercise for the reader to similarly construct the solutions for the other cases listed in Theorem \ref{theorem3} above.

\end{singlespace}
  \bibliography{SYMDW.bib}
  
  \bibliographystyle{JHEP}

\end{document}